\documentclass[preprint,12pt,authoryear]{elsarticle}

\usepackage{amssymb}
\usepackage{amsmath,stmaryrd}
\usepackage{amsthm}
\usepackage{mathtools}

    \usepackage[english]{babel}
    \usepackage[utf8]{inputenc}
\usepackage{color}
\usepackage{algorithm2e}
\usepackage{algpseudocode}
\usepackage{caption}
\usepackage{subcaption}

\DeclarePairedDelimiter\ceil{\lceil}{\rceil}
\DeclarePairedDelimiter\abs{\lvert}{\rvert}

\DeclarePairedDelimiter\floor{\lfloor}{\rfloor}

\newtheorem{theorem}{Theorem}

\journal{Computational Statistics and Data Analysis}

\begin{document}

\begin{frontmatter}



\title{A streaming algorithm for bivariate empirical copulas}


\author[l1,l2]{Alastair Gregory }

\address[l1]{Lloyd's Register Foundation's Programme for Data-Centric Engineering,
Alan Turing Institute}
\address[l2]{Department of Mathematics, Imperial College London \fnref{lb3}}
\fntext[lb3]{Alastair Gregory, Department of Mathematics, Imperial College London, Exhibition Road, South Kensington, SW7 2AZ, a.gregory14@imperial.ac.uk.}

\begin{abstract}
Empirical copula functions can be used to model the dependence structure of multivariate data. The Greenwald and Khanna algorithm is adapted in order to provide a space-memory efficient approximation to the empirical copula function of a bivariate stream of data. A succinct space-memory efficient summary of values seen in the stream up to a certain time is maintained and can be queried at any point to return an approximation to the empirical bivariate copula function with guaranteed error bounds. An example then illustrates how these summaries can be used as a tool to compute approximations to higher dimensional copula decompositions containing bivariate copulas. The computational benefits and approximation error of the algorithm is theoretically and numerically assessed.
\end{abstract}


\end{frontmatter}


\section{Introduction}

Streaming data is found in many applications where data is acquired continuously. This characteristic, in addition to any space-memory constraints of the user, make such data a challenge for analyses. As data is acquired the analyser of the data must utilise it before the next piece of data is acquired and the entire stream cannot be stored. Therefore, given a particular statistical quantity of the data, a \textit{summary} of the data with respect to this quantity must be maintained throughout time. This summary is typically much smaller in size than the entire stream. The idea of this summary is to allow an approximation of the desired statistical quantity to be made at any time with only a single pass of the data.

Estimating the quantiles of a data stream is a popular example of such a statistical quantity \citep{Buragohain}. A host of studies \citep{Arandjelovic,Greenwald,Munro, Manku} propose methods to construct succinct summaries of univariate data that can be queried at any time to obtain approximate quantiles within a guaranteed error bound $\epsilon$ (e.g. $\epsilon$-approximate quantile summaries). However, data is rarely univariate. Copula functions (empirical) are a natural way to model the dependencies between multiple streams of data. This paper adapts the aforementioned Greenwald and Khanna algorithm \citep{Greenwald} to construct an alternative bivariate data summary, returning queries to the empirical copula function with guaranteed error bounds. Whilst the paper doesn't directly extend the summary to higher dimensions, one can construct models of dependence for such multidimensional data using sets of pair-wise copulas \citep{Aas, Mazo}.
Therefore, approximations to such a copula can be found by using the $\epsilon'$-accurate bivariate copula functions considered here.

This work is related to other studies that also consider the construction of summaries for multidimensional data. These summaries have been used to query multidimensional ranks and ranges \citep{Hershberger, Suri, Yiu}. Querying multidimensional ranges, such as a rectangle of points on the plane, is analogous to finding empirical copulas, only considering the actual data points on the plane rather than the marginal quantiles. This is where our motivation differs to that of \cite{Suri} and \cite{Hershberger}. On this note, another closely related piece of literature to the scope of this paper is that of \cite{Xiao} which considers the online computation of pair-wise nonparametric correlations. However, this doesn't provide any theoretical error bounds on the summarized statistical approximations.

Due to the vast range of industries that use copulas to model dependent data, this application of copula models to streaming data is an important contribution to the data science community. The paper is structured as follows. A background on empirical copulas is given in the next section. In Sec.~\ref{sec:algorithm}, an algorithm to construct the summary used to obtain approximations of empirical copula functions is presented. This is followed by a theoretical and numerical assessment of the approximation from the algorithm in Secs.~\ref{sec:error} and \ref{sec:numerics} respectively. Section~\ref{sec:higherdim} gives an example of how higher dimensional copulas framed as sets of bivariate copulas can be approximated using the $\epsilon'$-approximate copulas presented in this paper. A discussion then concludes the paper.

\section{Copulas}

Copulas represent a joint probability distribution of a multidimensional random variable, and therefore can capture the dependence structure between components. The joint distribution is such that the marginal probability distributions of each component are uniform. Suppose we have two random variables $X_{(1)}\in \mathbb{R}$ and $X_{(2)}\in \mathbb{R}$, with marginal cumulative distribution functions (CDF) $F_{X_{(1)}}(x_{(1)})=P(X_{(1)} \leq x_{(1)})$ and $F_{X_{(2)}}(x_{(2)})=P(X_{(2)} \leq x_{(2)})$ respectively\footnote{For now we will only consider bivariate copulas. Then later in the paper, higher dimensional copulas will be considered.}. Then the copula function $C(u_1,u_2)$ is defined by
\begin{equation}
C(u_1,u_2)=F_{X_{(1)},X_{(2)}}\left(F^{-1}_{X_{(1)}}(u_1),F^{-1}_{X_{(2)}}(u_2)\right),
\label{equation:copula}
\end{equation}
where $(u_1,u_2) \in [0,1]^2$ and $F_{X_{(1)},X_{(2)}}(x_{(1)},x_{(2)})=P(X_{(1)} \leq x_{(1)}, X_{(2)} \leq x_{(2)})$ is the joint CDF of $X_{(1)}$ and $X_{(2)}$. Here $F^{-1}_{X_{(1)}}(u_{1})$ and $F^{-1}_{X_{(2)}}(u_2)$ are the inverse marginal CDFs (quantile functions). In the case where there does not exist unique values $x_{(1)} \in \mathbb{R}$ and $x_{(2)}\in \mathbb{R}$ that satisfy  $F_{X_{(1)}}(x_{(1)})=u_{1}$ and $F_{X_{(2)}}(x_{(2)})=u_2$, generalized inverse CDFs are used, where these are defined by
$$
\inf_{x_{(1)}  \in \mathbb{R}}F_{X_{(1)}}(x_{(1)}) \geq u_1 \quad \text{and} \quad \inf_{x_{(2)}\in\mathbb{R}}F_{X_{(2)}}(x_{(2)}) \geq u_2
$$
respectively \citep{Charpentier}.
There exist families of analytical copulas such as the Gaussian copula and Archimedean copulas, which can be fit to data streams $\big\{x^i_{(1)}, x^i_{(2)}\big\}_{i=1}^{n} \in \mathbb{R}^{2 \times n}$ where $x^i_{(1)} \sim X_{(1)}$ and $x_{(2)}^i \sim X_{(2)}$. Typically, this involves estimating the parameters within the copula using the data. For example, the Gaussian copula between $X_{(1)}$ and $X_{(2)}$ is given by
$$
C_{\text{Gaussian}}(u_1,u_2)= \Phi_C \left(\Phi^{-1}_{X_{(1)}}(u_1), \Phi^{-1}_{X_{(2)}}(u_2) \right),
$$
where $\Phi_C$ is the joint CDF corresponding to the Gaussian distribution $\mathcal{N}(0, \Sigma)$, $\Phi^{-1}_{X_{(1)}}$ is the inverse marginal Gaussian CDF of $X_{(1)}$ and $\Phi_{X_{(2)}}^{-1}$ is the inverse marginal Gaussian CDF of $X_{(2)}$. Here $\Sigma$ is the covariance matrix between $X_{(1)}$ and $X_{(2)}$ and can be estimated directly from the data. The mean and variance for the marginal Gaussian CDF's can also be estimated from the data stream $\big\{x^i_{(1)}, x^i_{(2)}\big\}_{i=1}^{n}$. This would be a suitable copula model to use if one knew the dependence structure between $X_{(1)}$ and $X_{(2)}$ to be Gaussian.


\subsection{Empirical copulas}
\label{sec:empiricalcopula}
For many data sets one wishes to compute an empirical copula where the dependence structure is unknown in advance. This empirical copula is based on concordant and discordant ranks of data points and therefore is linked to Kendall Tau correlation. Suppose $\emph{1}_{(x \leq y)}$ is the indicator function, taking the value of 1 if $x \leq y$, and 0 if $x > y$. Also let $\big\{\tilde{x}^{i}\big\}_{i=1}^{n}$ be the order statistics (ranked data) of the data stream $\big\{x^{i}\big\}_{i=1}^{n}\in \mathbb{R}^{n}$, such that $\tilde{x}^{1}<\tilde{x}^{2}<...<\tilde{x}^{n}$. An empirical copula \citep{Deheuvels1} of the bivariate data stream $\big\{x_{(1)}^{i},x_{(2)}^{i}\big\}_{i=1}^{n} \in \mathbb{R}^{2 \times n}$ is given by
\begin{equation}
\hat{C}(u_1,u_2) = \frac{1}{n}\sum^{n}_{i=1}\prod^{2}_{j=1}\emph{1}_{\left(x_{(j)}^{i}\leq \tilde{x}_{(j)}^{\ceil{u_jn}}\right)}.
 \label{equation:copula_estimate}
\end{equation}
The mean product of indicator functions in (\ref{equation:copula_estimate}) approximates the probability $F_{X_{(1)},X_{(2)}}(\tilde{x}_{(1)}^{\ceil{u_1n}},\tilde{x}_{(2)}^{\ceil{u_2n}})=P(X_{(1)} \leq \tilde{x}_{(1)}^{\ceil{u_1 n}}, X_{(2)} \leq \tilde{x}_{(2)}^{\ceil{u_2 n}})$. Further, the empirical copula approximates (\ref{equation:copula}) by using empirical inverse marginal CDFs as an approximation to $F_{X_{(k)}}^{-1}(u)$,
\begin{equation}
\hat{F}_{n,(k)}^{-1}(u) :=\tilde{x}_{(k)}^{\ceil{un}}, \qquad k=\{1,2\}.
\label{equation:empiricalquantile}
\end{equation}
This copula weakly converges (with the number of samples $n$) to the true underlying dependence structure between the two components of the data stream \citep{Deheuvels1}. There are a variety of different approximations to the quantile function $F^{-1}_{X_{(k)}}(u_k)$, for $k=\{1,2\}$ \citep{Ma}. A commonly used one shown in (\ref{equation:empiricalquantile}) is obtained by the piecewise constant function of the order statistics, $\hat{F}_{n,(k)}^{-1}(u)=\tilde{x}^{j}$, for $(j-1)/n < u \leq j/n$. The ceiling function in (\ref{equation:empiricalquantile}) is used to construct this piecewise function by noting that $\ceil{(j-1)+\delta}/n=j/n$ for $\delta \in (0,1]$. To simplify the analysis later on in the paper, let $I = \left\{j \in [1,n]; x_{(1)}^j \leq \tilde{x}^{\ceil{u_1 n}}_{(1)}\right\}$, and let $n_1=\abs{I}$. The dependence of $I$ on $u_1$ (taking the empirical inverse CDF of $X_{(1)}$ into account) allows (\ref{equation:copula_estimate}) to be expressed as
\begin{equation}
\hat{C}(u_1,u_2)=
\frac{1}{n}\sum^{n_{1}}_{i=1}\emph{1}_{\left(x_{(2)}^{I(i)} \leq \tilde{x}_{(2)}^{\ceil{u_2n}}\right)},
\end{equation}
where $I(i)$ is the $i$'th element in the set $I$. As with the inverse CDF, one can approximate the CDF empirically via
$$
\hat{F}_{n_1,(2)}(y)=\frac{1}{n_1}\sum^{n_1}_{i=1}\emph{1}_{\left(x_{(2)}^{I(i)} \leq y\right)}.
$$
Then one can state
\begin{equation}
\hat{C}(u_1,u_2)=\frac{n_1}{n}\hat{F}_{n_1,(2)}\left(\tilde{x}_{(2)}^{\ceil{u_2n}}\right)=\frac{n_1}{n}\hat{F}_{n_1,(2)}\left(\hat{F}^{-1}_{n,(2)}(u_2)\right).
\label{equation:newformcopula}
\end{equation}
This form of (\ref{equation:copula_estimate}) frames the expression in terms of empirical CDFs and inverse CDFs. The problem that this paper considers is updating (\ref{equation:newformcopula}) when elements are continuously added to the stream. The next section considers bivariate copulas in this streaming data scenario, and proposes a methodology to approximate bivariate empirical copulas for such data.

\section{Bivariate copulas for streaming data}

In the streaming data scenario, one does not wish to store the entire stream of data. Therefore, the estimation of copulas in this setting has to operate in an online manner. In the case of parametric copulas, a suitable approach would be to iteratively update the parameters within the copula model. There are some cases when this estimation would be exactly equivalent to that of computing them over the entire stream at once. For example, there are several parameters to estimate within the Gaussian copula model considered earlier. These are the mean and (population) standard deviation of each of the marginals, as well as the covariance between $X_{(1)}$ and $X_{(2)}$. The means can easily be updated as a new element is added to the stream $(x_{(1)}^{n+1},x_{(2)}^{n+1})$ by
$$
\mu^{n+1}_{(k)}=\frac{1}{n+1}\left(n\mu^{n}_{(k)}+x_{(k)}^{n+1}\right) , \quad \mu^{n}_{(k)}=\frac{1}{n}\sum^{n}_{i=1}x_{(k)}^i,
$$
for $k=\big\{1,2\big\}$. One can also follow similar updates for the standard deviation and covariance. For many other parametric copulas, Kendall's Tau is used to estimate the parameters within the copula model. An online computation of Kendall's Tau is available in \cite{Xiao}. Therefore, online estimation of the parameters within parametric copulas could follow. On the other hand, in the case of empirical copulas one cannot feasibly store the entire data stream to compute the order statistics of $x_{(1)}^{i}$ and $x_{(2)}^{i}$. 
%
%
Therefore the methodology presented in the following sections can be implemented to iteratively maintain an approximation to the bivariate empirical copula, $\hat{C}(u_1,u_2)$, over the data stream.


\subsection{Bivariate empirical copulas for streaming data}
\label{sec:algorithm}

An approximation to the bivariate empirical copula can be maintained over the course of the data stream by carefully updating a particular data structure, typically referred to as a statistical summary. The data structure proposed in this section is similar to both those used in \cite{Suri} and \cite{Hershberger} for the estimation of multidimensional ranges in data streams.
As adopted in the former study, the data structure proposed in this work for a \textit{copula summary} stores multiple versions of the quantile summary that was used in \cite{Greenwald}: lists of certain values seen in a data stream $\big\{x^{i}\big\}_{i=1}^{n}$, where each value `covers' the empirical quantiles in (\ref{equation:empiricalquantile}) within a different range (e.g. $u \in [0,0.1]$). The size of these quantile ranges is dependent on the approximation error that the user prescribes. On this note, define an $\epsilon$-approximate quantile summary $Q$ as one that can be queried for the $u$-empirical quantile $\tilde{x}^{\ceil*{u n}}$, and return a value $\tilde{x}^{j}$, where $j \in [\ceil*{u n}-\epsilon n,\ceil*{u n}+\epsilon n]$. The next paragraph will describe how to construct the quantile summary $Q$, and then the paragraph that follows will discuss how another summary that approximates bivariate empirical copulas can be formed from multiple versions of the quantile summary.

\subsubsection{Quantile summary}

The quantile summary $Q$ is composed of $L$ tuples $(z_i, g^i, \Delta^i)$, for $i=1,...,L$. The values $z_i \in \big\{x^k\big\}_{k=1}^n$, where $z_1 \leq z_2 \leq \ldots \leq z_L$, are a selection of data points that have been seen in the data stream so far. The parameters $g^{i}$ and $\Delta^i$ in all tuples within the summary are required to infer the range of empirical quantiles that each element in the summary, $z_i$, covers. On this note, let $r_{min,Q}(z_i)$ and $r_{max,Q}(z_i)$ be the rank of the element in $\big\{x^k\big\}_{k=1}^n$ that corresponds to the minimum and maximum empirical quantiles covered by the summary value $z_i$ respectively. The parameters $g^i$ and $\Delta^i$ infer these ranks via the governing equations,
$$
r_{min,Q}(z_i)-r_{min,Q}(z_{i-1})=g^i, \qquad r_{max,Q}(z_i)-r_{min,Q}(z_i)=\Delta^i,
$$
with $r_{min,Q}(z_0)=0$.
The values of $r_{min,Q}(z_i)$ and $r_{max,Q}(z_i)$ are minimum and maximum bounds on the rank that the element $z_i$ took in the original stream. This means that the upper bound on the number of elements in the original stream between $z_{i-1}$ and $z_i$ is $g^i+\Delta^i-1$. The Greenwald and Khanna algorithm updates the quantile summary in a manner that guarantees that
\begin{equation}
r_{max,Q}(z_i)-r_{min,Q}(z_{i-1})=g^i+\Delta^i \leq 2\epsilon n,
\label{equation:quantilecondition}
\end{equation}
at all times. Due to this guarantee, it follows that a query of the rank of an element $y\in \big\{x^k\big\}_{k=1}^n$ in the original stream, where $z_{i-1} \leq y \leq z_{i}$, can be answered to within an $\epsilon n$ tolerance \citep{Greenwald}.

\subsubsection{Copula summary}

Now given a bivariate data stream $\big\{x_{(1)}^i,x_{(2)}^i\big\}_{i=1}^n$ the structure of the proposed copula summary, formed using multiple versions of the quantile summaries explained above, is now described. It starts by maintaining an $\epsilon$-approximate quantile summary, $S_{(1)}$, for the first components of the elements in the bivariate data stream $\big\{x_{(1)}^{i}\big\}_{i=1}^{n}$. Suppose this summary is $L$ elements long. The summary is composed from the following tuples: $(v_i,g_{(1)}^i,\Delta_{(1)}^i)$, for $i=1,...,L$. To accompany each of the elements in this summary are $L$ different $\epsilon$-approximate quantile (sub)summaries $S_{(2)}^{i}$ of length $L_i$, for $i=1,...,L$. Here, $v_i\in \big\{x_{(1)}^k\big\}_{k=1}^{n}$ is the first component of a data point seen in the stream so far. As aforementioned, the parameters $g_{(1)}^{i}$ and $\Delta_{(1)}^{i}$ enforce the range of quantiles that each element $v_i$ covers in the stream $\big\{x_{(1)}^k\big\}_{k=1}^{n}$. Finally, each $S_{(2)}^{i}$ is a quantile summary for the second component of a selection of the data points seen in the stream so far. These points will not in general correspond to points with the first component $v_i$ (i.e. the coupling between the two components of each point is lost), however it is permissible for the motivation of this paper. Each subsummary $S_{(2)}^i$ is formed of tuples $(w_j,g_{(2)}^{i,j},\Delta_{(2)}^{i,j})$, for $j=1,...,L_i$, where $w_j\in \big\{x_{(2)}^k\big\}_{k=1}^{n}$. Once again, the parameters $g_{(2)}^{i,j}$ and $\Delta_{(2)}^{i,j}$ work in the same way as $g_{(1)}$ and $\Delta_{(1)}$ in enforcing ranges of quantiles. In total then, there are $L+1$ different $\epsilon$-approximate quantile summaries stored. This data structure resembles a grid of the joint ranks of the data, and is analogous to the grid of quantiles used in \cite{Xiao} to this end. The collection of summaries $\big\{S_{(1)}, S_{(2)}^1,\ldots,S_{(2)}^L\big\}$ will henceforth be referred to as the copula summary. The following subsections will describe how this copula summary can be updated as further elements join the bivariate data stream and is used to answer empirical copula function queries to a particular error tolerance.


\subsection{Updating the copula summary}

Two operations (insert and combine) are used to maintain the standard $\epsilon$-approximate quantile summaries in \cite{Greenwald} when new elements are added to the data stream, whilst guaranteeing (\ref{equation:quantilecondition}). These can be modified to update the copula summary.

\subsubsection{Insert}
\label{sec:algorithminsert}

When an element $(x_{(1)}^{n+1},x_{(2)}^{n+1})$ gets added to the bivariate data stream $\big\{x_{(1)}^i,x_{(2)}^i\big\}_{i=1}^n$, a tuple $(x_{(1)}^{n+1},1,\Delta_{(1)}^*)$ gets added to the quantile summary $S_{(1)}$. Here the subsummary $S_{(2)}^*=\big\{(x_{(2)}^{n+1},1,0)\big\}$ also gets added to the copula summary. For more details on this operation see \ref{sec:appendixinsert}.

\subsubsection{Combine}
\label{sec:algorithmcombine}

After a particular number of elements get added to the copula summary (using the insert operation described in the previous section), it is necessary to combine and merge tuples within the summary. This means that the copula summary will be a succinct summary, and not storing every element in the data stream. In general, successive tuples will be merged into a single tuple if the range of quantiles they jointly cover, in either the first component summary $S_{(1)}$ or second component subsummaries $S_{(2)}^i$, is $\leq 2\epsilon n$, from (\ref{equation:quantilecondition}). This operation therefore makes sure that $S_{(1)}$ and $S_{(2)}^i$, for $i=1,...,L$, are $\epsilon$-approximate marginal quantile summaries for the first and second components of the data stream respectively. For more details on this operation see \ref{sec:appendixcombine}.

\subsection{Querying the copula summary}

\label{sec:copulaquery}

The copula summary can be updated after new elements are added to the bivariate data stream using the operations described in the previous section. Now the following section explains how this summary can be queried at any time to return an approximation to the empirical copula function. The section sequentially describes approximations to the different components of the empirical copula in (\ref{equation:newformcopula}). 
First recall that $S_{(1)}$ is an $\epsilon$-approximate quantile summary for the first component of the bivariate data stream. This means that one can query the summary $S_{(1)}$ for the $u_1$-quantile of $\big\{x_{(1)}^i\big\}_{i=1}^{n}$ and have an approximation $\tilde{x}^{j}_{(1)}$ returned, where $j \in [\ceil{u_1 n}-\epsilon n,\ceil{u_1 n}+\epsilon n]$ \citep{Greenwald}. For full details on how to implement such a query see \ref{sec:appendixquery}. Denote this query, an approximation to $\hat{F}^{-1}_{n,(1)}(u_1)=\tilde{x}^{\ceil{u_1 n}}_{(1)}$, by $\tilde{F}^{-1}_{n,(1)}(u_1)$.

Then let
\begin{equation}
\hat{n}_1=\sum^{E}_{i=1}\sum^{L_{i}}_{j=1}g_{(2)}^{i,j},
\label{equation:hatndefine}
\end{equation}
where
\begin{equation}
E = \text{arg} \min\big\{i \in \{1,\ldots,L\};\tilde{F}^{-1}_{n,(1)}(u_1) \geq v_i\big\}.
\label{equation:Edefine}
\end{equation}
Recall that $v_i$, for $i=1,\ldots,L$, are the elements in the summary $S_{(1)}$. The value $\hat{n}_1$ is an approximation to $n_1$, defined in Sec. \ref{sec:empiricalcopula}. Next we take advantage of the fact that multiple subsummaries $S_{(2)}^i$, for $i \in s \subset \big\{1,...,L\big\}$, can be merged into one $\epsilon$-approximate quantile summary $M(\big\{S_{(2)}^i\big\}_{i\in s})$ by using the methodology in \cite{GreenwaldMerge} and described in \ref{sec:merging}. In the present work this allows one to approximate the $u_2$-quantile of $\big\{x_{(2)}^i\big\}_{i=1}^{n}$ by querying the $\epsilon$-approximate quantile summary $M(S_{(2)}^{1},$ $\ldots,S_{(2)}^{L})$. Denote this query by $\tilde{F}^{-1}_{n,(2)}(u_2)$. Finally, one can also find an approximation to the empirical CDF $\hat{F}_{n_1,(2)}(y)$ that appears in the empirical copula approximation in (\ref{equation:newformcopula}) via an `inverse' summary query described in \cite{Lall} and \ref{sec:appendixinversequery}. Denote this inverse query on the merged $\epsilon$-approximate summary $M(S_{(2)}^{1},\ldots,S_{(2)}^{E})$ by $\tilde{F}_{\hat{n}_1,(2)}(y)$. Combining all of the different queries described above together, we have
\begin{equation}
\hat{C}_{S}(u_1,u_2)=\frac{\hat{n}_1}{n}\tilde{F}_{\hat{n}_1,(2)}\left(\tilde{F}^{-1}_{n,(2)}(u_2)\right),
\label{equation:copulasummaryintext}
\end{equation}
as the copula summary query and the approximation to the empirical copula $\hat{C}(u_1,u_2)$. This query is described in more detail in \ref{sec:appendixcopulaquery}. The next section provides a theoretical analysis of the error of this approximation.

\section{Error and efficiency analysis}
\label{sec:error}

This section provides a theoretical analysis on the error and efficiency of the approximation $\hat{C}_{S}(u_1,u_2)$. The bound on the error of this approximation away from (\ref{equation:copula_estimate}) is now stated and proved in the following theorems.

\begin{theorem}[Error bound]
Let $\hat{C}(u_1,u_2)$ be the empirical copula function of the bivariate stream of data $\big\{x_{(1)}^{i},x_{(2)}^{i}\big\}_{i=1}^{n} \in \mathbb{R}^{2 \times n}$ evaluated at $(u_1, u_2) \in [0,1]^2$. Also suppose that $\hat{C}_{S}(u_1,u_2)$ is as it is defined in (\ref{equation:copulasummaryintext}), then
$$
\abs*{\hat{C}_{S}(u_1,u_2)-\hat{C}(u_1,u_2)} \leq 5\epsilon.
$$
\end{theorem}

\begin{proof}
Note that
$$
\abs*{\hat{C}_{S}(u_1,u_2)-\hat{C}(u_1,u_2)} = \abs*{\frac{\hat{n}_{1}}{n}\tilde{F}_{\hat{n}_1,(2)}\left(\tilde{F}^{-1}_{n,(2)}(u_2)\right)-\frac{n_1}{n}\hat{F}_{n_1,(2)}\left(\hat{F}^{-1}_{n,(2)}(u_2)\right)},
$$
due to (\ref{equation:newformcopula}) and therefore can be split up into three contributing parts by the triangle inequality,
\begin{equation}
\begin{split}
\abs*{\hat{C}_{S}(u_1,u_2)-\hat{C}(u_1,u_2)} &\leq \underbrace{\abs*{\frac{\hat{n}_{1}}{n}\tilde{F}_{\hat{n}_1,(2)}\left(\tilde{F}^{-1}_{n,(2)}(u_2)\right)-\frac{\hat{n}_{1}}{n}\hat{F}_{\hat{n}_1,(2)}\left(\tilde{F}^{-1}_{n,(2)}(u_2)\right)}}_{\text{(A)}} \\
\quad &+\underbrace{\abs*{\frac{\hat{n}_1}{n}\hat{F}_{\hat{n}_1,(2)}\left(\tilde{F}^{-1}_{n,(2)}(u_2)\right)-\frac{\hat{n}_1}{n}\hat{F}_{\hat{n}_1,(2)}\left(\hat{F}^{-1}_{n,(2)}(u_2)\right)}}_{\text{(B)}}\\
\quad &+ \underbrace{\abs*{\frac{\hat{n}_{1}}{n}\hat{F}_{\hat{n}_1,(2)}\left(\hat{F}^{-1}_{n,(2)}(u_2)\right)-\frac{n_{1}}{n}\hat{F}_{n_1,(2)}\left(\hat{F}^{-1}_{n,(2)}(u_2)\right)}}_{\text{(C)}}.
\end{split}
\end{equation}
The proof is now split up into sub-theorems (Theorems \ref{theorem:A}, \ref{theorem:B} and \ref{theorem:C}) corresponding to the three parts above.
\end{proof}

The error can therefore be framed as taking a sum of the errors from steps (3) and (4) in Sec.~\ref{sec:algorithmcombine} (A and B) in addition to those from steps (1) in Sec.~\ref{sec:algorithmcombine} (C). Each of these contributing errors are now bounded.

\begin{theorem}[Error bound on (A)]
Suppose $u_2 \in [0,1]$, then
$$
\abs*{\frac{\hat{n}_{1}}{n}\tilde{F}_{\hat{n}_1,(2)}\left(\tilde{F}^{-1}_{n,(2)}(u_2)\right)-\frac{\hat{n}_{1}}{n}\hat{F}_{\hat{n}_1,(2)}\left(\tilde{F}^{-1}_{n,(2)}(u_2)\right)} \leq 3\epsilon.
$$
\label{theorem:A}
\end{theorem}

\begin{proof}
This is the guaranteed error bound for inversely querying an $\epsilon$-approximate summary, from \cite{Lall}.
\end{proof}

\begin{theorem}[Error bound on (B)]
Suppose $u_2 \in [0,1]$, then
$$
\abs*{\frac{\hat{n}_1}{n}\hat{F}_{\hat{n}_1,(2)}\left(\tilde{F}^{-1}_{n,(2)}(u_2)\right)-\frac{\hat{n}_1}{n}\hat{F}_{\hat{n}_1,(2)}\left(\hat{F}^{-1}_{n,(2)}(u_2)\right)} \leq \epsilon.
$$
\label{theorem:B}
\end{theorem}

\begin{proof}
Let $\xi=\ceil{u_2n}$. Then suppose the element returned by querying the $\epsilon$-approximate summary $M(S_{(2)}^{1},...,S_{(2)}^{L})$ for the $u_2$-quantile is $\tilde{x}_{(2)}^{\gamma}$. Therefore $\abs*{\xi-\gamma} \leq \epsilon n$. Now define
$$
\hat{I}=\big\{j \in [1,n]; x_{(1)}^j \leq \tilde{x}^{\hat{n}_1}_{(1)}\big\}.
$$
Recall from Sec.~\ref{sec:copulaquery} that $\hat{n}_1\hat{F}_{\hat{n}_1,(2)}(y)$ is simply the count of all samples in $\big\{x_{(2)}^{\hat{I}(i)}\big\}_{i=1}^{\hat{n}_1}$ less than or equal to $y$. For $y=\tilde{x}_{(2)}^{\xi}$, let this count be denoted by $\xi_R$. For $y=\tilde{x}_{(2)}^{\gamma}$, let this count be denoted by $\gamma_R$. As $\hat{I} \subset \big\{1,....,n\big\}$, we have $\abs*{\xi_R - \gamma_R} \leq \epsilon n$ also. Therefore
$$
\abs*{\hat{n}_1\hat{F}_{\hat{n}_1,(2)}\left(\tilde{F}^{-1}_{n,(2)}(u_2)\right)-\hat{n}_1\hat{F}_{\hat{n}_1,(2)}\left(\hat{F}^{-1}_{n,(2)}(u_2)\right)} \leq \epsilon n,
$$
and finally
$$
\abs*{\frac{\hat{n}_1}{n}\hat{F}_{\hat{n}_1,(2)}\left(\tilde{F}^{-1}_{n,(2)}(u_2)\right)-\frac{\hat{n}_1}{n}\hat{F}_{\hat{n}_1,(2)}\left(\hat{F}^{-1}_{n,(2)}(u_2)\right)} \leq \epsilon.
$$
\end{proof}

\begin{theorem}[Error bound on (C)]

Suppose $u_2 \in [0,1]$, then
$$
\abs*{\frac{\hat{n}_{1}}{n}\hat{F}_{\hat{n}_1,(2)}\left(\hat{F}^{-1}_{n,(2)}(u_2)\right)-\frac{n_{1}}{n}\hat{F}_{n_1,(2)}\left(\hat{F}^{-1}_{n,(2)}(u_2)\right)} \leq \epsilon.
$$
\label{theorem:C}
\end{theorem}

\begin{proof}
Recall the definition of $\hat{I}$ from the previous proof. Let $A=\big\{x_{(2)}^{\hat{I}(i)}\big\}_{i=1}^{\hat{n}_1}$. Recall from Sec.~\ref{sec:empiricalcopula} that $B=\big\{x_{(2)}^{I(i)}\big\}_{i=1}^{n_1}$ are the $n_1$ elements that have corresponding values $\big\{x_{(1)}^{I(i)}\big\}_{i=1}^{n_1}$ with ranks less than or equal to $\ceil{u_1n}$ in the original stream. We assume without loss of generality that if $n_1 < \hat{n}_1$ then $B\subset A$, and vice-versa if $\hat{n}_1<n_1$. Define $\xi$ to be the count of all elements in $B$ that are less than or equal to $\hat{F}^{-1}_{n,(2)}(u_2)$, which is equivalent to $n_1 \hat{F}_{n_1,(2)}\left(\hat{F}^{-1}_{n,(2)}(u_2)\right)$. Then by the fact that $S_{(1)}$ is an $\epsilon$-approximate quantile summary of $\big\{x_{(1)}^{i}\big\}_{i=1}^{n}$, the count of all elements in $A$ that are less than or equal to $\hat{F}^{-1}_{n,(2)}(u_2)$, which is equivalent to $\hat{n}_1 \hat{F}_{\hat{n}_1,(2)}\left(\hat{F}^{-1}_{n,(2)}(u_2)\right)$, is within the interval $[\xi-\epsilon n,\xi+\epsilon n]$. Therefore,
\begin{equation}
\abs*{\frac{\hat{n}_{1}}{n}\hat{F}_{\hat{n}_1,(2)}\left(\hat{F}^{-1}_{n,(2)}(u_2)\right)-\frac{n_{1}}{n}\hat{F}_{n_1,(2)}\left(\hat{F}^{-1}_{n,(2)}(u_2)\right)} \leq \abs*{\frac{\xi \pm \epsilon n}{n}-\frac{\xi}{n}} = \frac{\epsilon n}{n} = \epsilon.
\end{equation}
\end{proof}

The main benefit of this algorithm is that in streams of bivariate data acquired continuously one can compute the approximation to the empirical copula function, bounded in the theorems above, by maintaining a succinct summary of the data. It does this by storing a separate quantile summary $S_{(2)}^{i}$, for each $i=1,...,L$ elements in a single quantile summary $S_{(1)}$, all of which are $\epsilon$-approximate. From \cite{Greenwald}, the length of an $\epsilon$-approximate summary constructed using the insert and combine operations discussed in Sec. \ref{sec:algorithminsert} and \ref{sec:algorithmcombine} is at the worst-case $L=\mathcal{O}\left(\frac{1}{\epsilon} \log(\epsilon n)\right)$. Therefore for the algorithm considered in this paper, the worst-case number of tuples stored in the copula summary at any one time is $\mathcal{O}\left(\frac{L}{\epsilon}\log(\epsilon n)\right)=\mathcal{O}\left(\frac{1}{\epsilon^2}\log(\epsilon n)^2\right)$. This is the same complexity as the queries of very similar data structures in \cite{Suri} and \cite{Hershberger} obtained for multidimensional range counting. It is worth noting, as seen in \cite{Greenwald}, the space-memory of a single quantile summary is much better than this worst-case in practice. In many cases, such as when one implements the combine operation after an element is added to a single quantile summary rather than after every $1/\floor*{2\epsilon}$ steps, the space-memory used is independent of $n$.

The $\epsilon$-approximate quantile summaries utilised in this paper are uniformly accurate across all quantiles $u \in [0,1]$. It is possible to adjust the condition in (\ref{equation:quantilecondition}) to allow for certain quantile approximations to be more accurate than others using an $\epsilon$-approximate quantile summary \citep{Cormode}. Commonly the high quantiles $1-u$, $1-u^2$, $1-u^3$, \ldots , $1-u^k$ are of interest. These have been referred to as \textit{biased quantile approximations}. One could extend this methodology to the copula summaries presented in this paper by adjusting the insert and combine operations in Sec. \ref{sec:algorithminsert} and \ref{sec:algorithmcombine} acting on $S_{(1)}$ and $S_{(2)}^1,\ldots,S_{(2)}^{L}$. This is of particular relevance to the field of copulas, as one is often interested in computing the tail dependence (coefficient) between two random variables \citep{Schmidt}.
Based on the analysis above, it is apparent that a direct extension of this algorithm to a higher dimension $d$, where such a data structure uses $\mathcal{O}\left(\frac{1}{\epsilon^d}\log(\epsilon n)^d\right)$ space-memory, would be infeasible. This is noted in \cite{Hershberger} for a very similar data structure. However, the next section gives an explanation of how one may model the dependence structure of high dimensional data streams by utilising these bivariate copula summaries.

\section{Higher dimensional copulas}
\label{sec:higherdim}

So far this paper has only discussed bivariate copulas for two streams of data. However, this section now gives a brief example of how the approximations from bivariate copula summaries can be used to construct approximations to higher dimensional copulas. It is well known that higher dimensional copulas $C(u_1,\ldots,u_d)$, for $d>2$, can be framed as decompositions containing sets of $d(d-1)/2$ (conditional) bivariate pair-copulas \citep{Aas, Mazo, Bedford} (e.g. pair-copula construction). This corresponds to each of the $d$ components being a node in a fully connected dependence graph. For high dimensions, there are many different decompositions for the copula $C(u_1,\ldots,u_d)$, and therefore often \textit{vines} are a useful tool. Given a copula modelling 5 random variables, there are 240 possible decompositions. For more information on these, turn to \cite{Aas}. Whilst these decompositions provide complexity in deriving conditional copula densities, they offer a very adaptable framework for constructing higher dimensional copulas. Denote $u_{D}=F_{X_{(D_{1})}, \ldots, X_{(D_{\abs{D}})}}(x_{(D_1)},\ldots,x_{(D_{\abs{D}})})$, where $D \subset \big\{1,\ldots,d\big\}$, to be the (joint) distribution function for the random variables $X_{(D_{1})}, \ldots, X_{(D_{\abs{D}})}$, and similarly $u_{i|\cdot}=F_{X_{(i)}|\cdot}(x_{(i)}|\cdot)$ to be the conditional distribution function of $X_{(i)}$. A possible decomposition (known as the $D$-vine) of $C(u_1,\ldots,u_d)$ is
\begin{equation}
C(u_1,\ldots,u_d)=\prod^{d-1}_{j=1}\prod^{d-j}_{i=1}C\left(u_{i|w_{i,j}},u_{(i+j)|w_{i,j}};u_{w_{i,j}}\right),
\label{equation:decomposition}
\end{equation}
where
\begin{equation*}
w_{i,j}=
\begin{cases}
\big\{i+1,\ldots,i+j-1\big\}, &j \geq 2 \\
\big\{\big\}, & j < 2.
\end{cases}
\end{equation*}
Here it is common practice to make the simplifying assumption that the conditional copulas $C\left(u_{i|w_{i,j}},u_{(i+j)|w_{i,j}};u_{w_{i,j}}\right)$ are constant over $u_{w_{i,j}}$. Let $w_{i,j}'=w_{i,j} \setminus (i+j)$, then the conditional $u_{i|w_{i,j}}$ is given by (for $j \geq 2$),
\begin{equation*}
\begin{split}
u_{i|w_{i,j}}&:=h\big(u_{i|w_{i,j}'}|u_{(i+j)|w_{i,j}'}\big)=\frac{\partial C\big(u_{i|w_{i,j}'},u_{(i+j)|w_{i,j}'}\big)}{\partial u_{(i+j)|w_{i,j}'}}\\
\quad &=\int_0^{u_{i|w_{i,j}'}} C\big(v,u_{(i+j)|w_{i,j}'}\big)dv.
\end{split}
\end{equation*}
In the same way let $w_{i,j}'=w_{i,j}\setminus i$, and $u_{(i+j)|w_{i,j}}=h\big(u_{(i+j)|w_{i,j}'}|u_{i|w_{i,j}'}\big)$. Therefore each conditional copula density can be framed as a recursion of the expressions above as elements are removed from $w_{i,j}$, until $w_{i,j}=\big\{\big\}$ and $C\left(u_{i|w_{i,j}},u_{(i+j)|w_{i,j}}\right)=C(u_i,u_{(i+j)})$ is an unconditional bivariate copula. For example, using the framework above a possible decomposition of the copula $C(u_1,u_2,u_3)$ is given by
\begin{equation}
C(u_1,u_2,u_3)=C(u_1,u_2)C(u_2,u_3)C(u_{1|2},u_{3|2}).
\label{equation:threepcc}
\end{equation}
Now suppose we have a data stream $\big\{x_{(1)}^k,\ldots,x_{(d)}^k\big\}_{k=1}^n$ for the random variables $X_{(1)},\ldots,X_{(d)}$. In the case where the bivariate copulas used in the decomposition above are empirical copulas, the unconditional bivariate copulas, e.g. $\hat{C}(u_i,u_{i+j})$, can be simply computed via (\ref{equation:newformcopula}). The conditional pair copulas, e.g. $\hat{C}(u_{i|w_{i,j}},u_{(i+j)|w_{i,j}})$, are required to be computed using `pseudo-observations' $\big\{\hat{u}_{i|w_{i,j}}^k,\hat{u}_{(i+j)|w_{i,j}}^k\big\}_{k=1}^{n}$ \citep{Nagler}. These are given by (where $w_{i,j}'=w_{i,j}\setminus (i+j)$),
\begin{equation}
\hat{u}_{i|w_{i,j}}^k=\hat{h}\left(\hat{u}_{i|w_{i,j}'}^k|\hat{u}_{(i+j)|w_{i,j}'}^k\right),
\label{equation:pseudoobservations}
\end{equation}
where
\begin{equation}
\hat{h}\left(\hat{u}_{i|w_{i,j}'}^k|\hat{u}^k_{(i+j)|w_{i,j}'}\right)=\int_0^{\hat{u}_{i|w_{i,j}'}^k} \hat{C}\big(v,\hat{u}_{(i+j)|w_{i,j}'}^k\big)dv,
\label{equation:empiricalh}
\end{equation}
and vice-versa for $\hat{u}_{(i+j)|w_{i,j}}^k$ where $w_{i,j}'=w_{i,j}\setminus i$. These integrals can be computed numerically (e.g. trapezoidal rule). From these pseudo-observations, one can compute the `pseudo-data'
\begin{equation}
\big\{x_{(i)|w_{i,j}}^k,x_{(i+j)|w_{i,j}}^k\big\}_{k=1}^{n}=\big\{\hat{F}^{-1}_{n,(i)}\big(\hat{u}_{i|w_{i,j}}^k\big), \hat{F}^{-1}_{n,(i+j)}\big(\hat{u}_{(i+j)|w_{i,j}}^k\big)\big\}_{k=1}^{n},
\label{equation:pseudodata}
\end{equation}
using the empirical inverse marginal CDFs in (\ref{equation:empiricalquantile}). Finally the pseudo-data $\big\{x_{(i)|w_{i,j}}^k,x_{(i+j)|w_{i,j}}^k\big\}_{k=1}^{n}$ can be used to construct the (conditional) empirical copula $\hat{C}(u_{i|w_{i,j}},u_{(i+j)|w_{i,j}})$ via (\ref{equation:newformcopula}). All conditional empirical copulas are then recursively computed to obtain all components of the decomposition in (\ref{equation:decomposition}); on this note let a decomposition of the higher dimensional empirical copula be given by
\begin{equation}
\hat{C}(u_1,\ldots,u_d)=\prod^{d-1}_{j=1}\prod^{d-j}_{i=1}\hat{C}\left(u_{i|w_{i,j}},u_{(i+j)|w_{i,j}}\right).
\label{equation:empiricaldecomposition}
\end{equation}
An idea for an approximation to this decomposition of empirical pair copulas in the streaming data context, using the bivariate copula summaries presented in this paper, will be explained in the next section.

\bigskip

\subsection{Computing higher dimensional empirical copula approximations from copula summaries}

For the continuous data stream $\big\{x_{(1)}^k,\ldots,x_{(d)}^k\big\}_{k=1}^n$, one would like to maintain an approximation to the decomposition of higher dimensional empirical copulas in (\ref{equation:empiricaldecomposition}) described in the last section. For the scope of this paper, we will restrict ourselves to the case where the bivariate copula summaries proposed in this paper are only kept for the unconditional bivariate copulas in the decomposition. Then only the last $n_{\text{query}}$ data points in the stream are used to construct the pseudo-observations empirically, e.g. $\hat{u}^k_{i|w_{i,j}}$. These pseudo-observations are then used to construct standard empirical conditional pair copulas as in (\ref{equation:pseudodata}). This is so that the integrals in (\ref{equation:empiricalh}) can be computed using the same copula density for all pseudo-observations (and for a particular value of $i$ and $j$). It is also assumed that $n_{\text{query}}$ is small enough to store this buffer of data in space-memory temporarily.
A simple algorithm for maintaining an approximation to the decomposition in (\ref{equation:empiricaldecomposition}) is shown below.

\begin{itemize}
\item[(1)] Construct bivariate copula summaries, $\hat{C}_{S}(u_{i},u_{(i+1)})$, for the unconditional bivariate copulas over the entire data stream (when $j=1$ and $w_{i,j}=\{\}$) using Sec. \ref{sec:algorithm}.

\item[(2)] Let $\big\{\hat{u}^k_{i},\hat{u}^k_{(i+1)}\big\}_{k=1}^{n_{\text{query}}}=\big\{\tilde{F}_{n,(i)}(x_{(i)}^k),\tilde{F}_{n,(i+1)}(x_{(i+1)}^k)\big\}_{k=n-n_{\text{query}}+1}^{n}$, where $\tilde{F}_{n,(i)}$ is the inverse query for the $\epsilon$-approximate quantile summary $S_{(1)}$ in the copula summary $\hat{C}_S(u_i,u_{(i+1)})$. In the same way $\tilde{F}_{n,(i+1)}$ is the inverse query for the $\epsilon$-approximate quantile summary $M(S_{(2)}^1,\ldots,S_{(2)}^L)$ in the same copula summary. Then compute the pseudo-observations
$$
\hat{u}^k_{i|(i+1)}=\int_0^{\hat{u}^k_{i}}\hat{C}_S\big(v,u_{(i+1)}^k\big)dv,
$$
and
$$
\hat{u}^k_{(i+1)|i}=\int_0^{\hat{u}^k_{(i+1)}}\hat{C}_S\big(u_{i}^k,v\big)dv
$$
for $k=1,\ldots,n_{\text{query}}$. Then compute the pseudo-data
\begin{equation}
\big\{x_{(i)|w_{i,j}}^k,x_{(i+j)|w_{i,j}}^k\big\}_{k=1}^{n_{\text{query}}}=\big\{\tilde{F}^{-1}_{n,(i)}\big(\hat{u}_{i|w_{i,j}}^k\big), \tilde{F}^{-1}_{n,(i+j)}\big(\hat{u}_{(i+j)|w_{i,j}}^k\big)\big\}_{k=1}^{n_{\text{query}}},
\label{equation:pseudodatasummary}
\end{equation}
where $\tilde{F}^{-1}_{n,(i)}$ is the quantile query of $S_{(1)}$ in the copula summary $\hat{C}_{S}\left(u_i,u_{(i+j)}\right)$ and $\tilde{F}^{-1}_{n,(i+j)}$ is the quantile query of $M(S_{(2)}^1,\ldots,S_{(2)}^L)$ in the same copula summary. This pseudo-data is used to construct the conditional empirical copulas $\hat{C}\left(u_{i|w_{i,j}},u_{(i+j)|w_{i,j}}\right)$ via (\ref{equation:newformcopula}).

\item[(3)] Recursively compute the pseudo-observations using (\ref{equation:pseudoobservations}) and (\ref{equation:empiricalh}), whilst computing the pseudo-data in (\ref{equation:pseudodatasummary}) to construct
conditional empirical copulas $\hat{C}\left(u_{i|w_{i,j}},u_{(i+j)|w_{i,j}}\right)$ at each step.

\item[(4)] Find the product of all bivariate copula approximations,
\begin{equation}
\hat{C}_{S}(u_1,\ldots,u_d)=\left(\prod^{d-1}_{i=1}\hat{C}_{S}\left(u_{i},u_{(i+1)}\right)\right)\left(\prod^{d-1}_{j=2}\prod^{d-j}_{i=1}\hat{C}\left(u_{i|w_{i,j}},u_{(i+j)|w_{i,j}}\right)\right).
\label{equation:summarydecomposition}
\end{equation}
\end{itemize}

The error, with respect to (\ref{equation:empiricaldecomposition}) whilst only using the last $n_{\text{query}}$ data points in the computation of the pseudo-observations and pseudo-data in (\ref{equation:pseudoobservations}), (\ref{equation:empiricalh}) and (\ref{equation:pseudodata}), will stay proportional to the cumulative error from the unconditional bivariate copula summaries over time. This is because only a fixed number of points are used in the construction of the conditional pair copulas. In practice, the higher the value of $n_{\text{query}}$, the more accurate the conditional pair copulas will be. A demonstration of this construction, for the three-dimensional decomposition in (\ref{equation:threepcc}), will be given in the next section.

\section{Numerical demonstration}
\label{sec:numerics}

\begin{figure}[t!]
\centering
\begin{subfigure}{0.7\textwidth}
  \centering
  \includegraphics[width=0.9\textwidth]{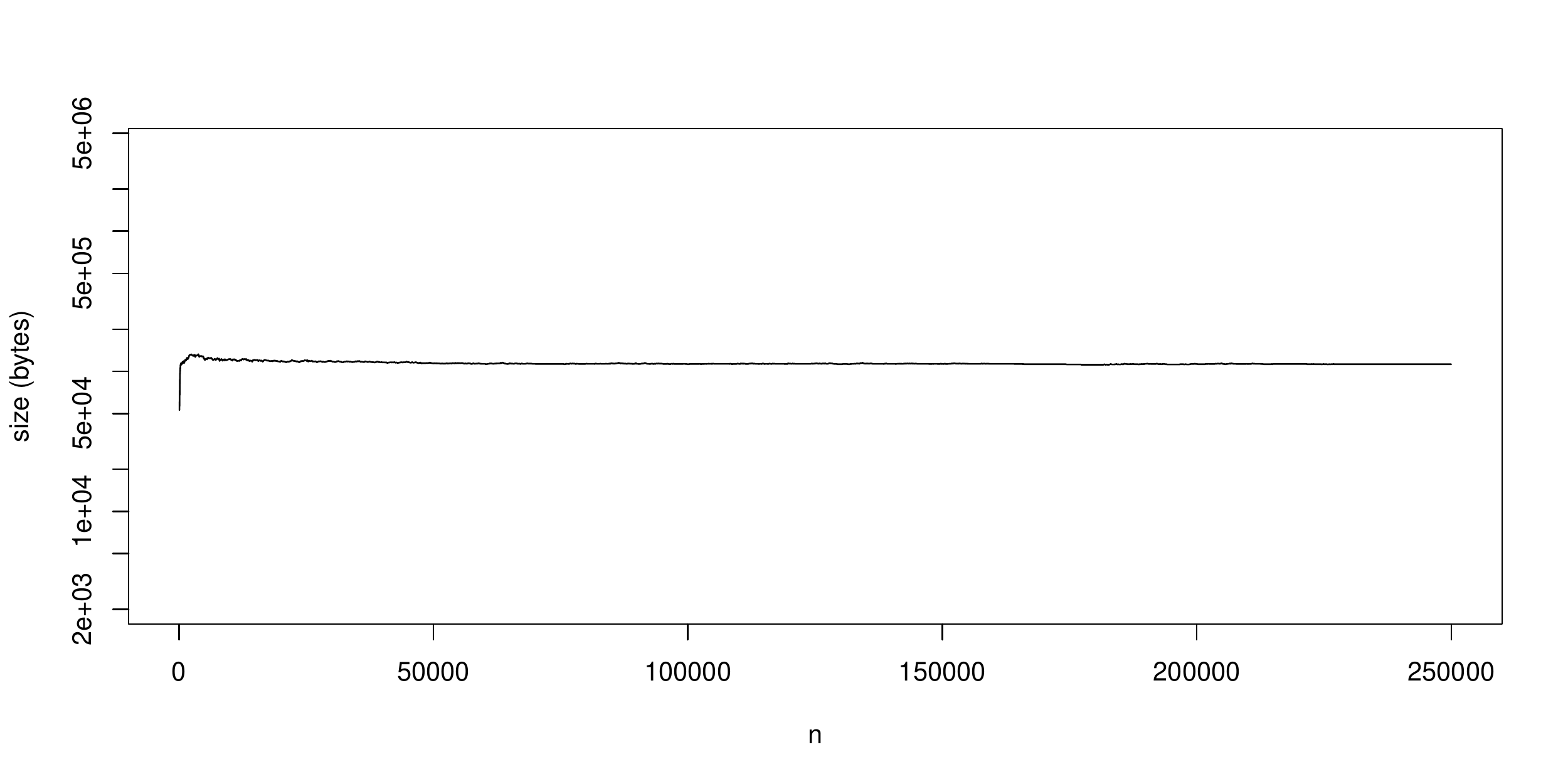}
  \caption{}
  \label{fig:size_copula_summary}
\end{subfigure}%
\vspace{5mm}
\begin{subfigure}{0.7\textwidth}
\centering
  \includegraphics[width=0.9\textwidth]{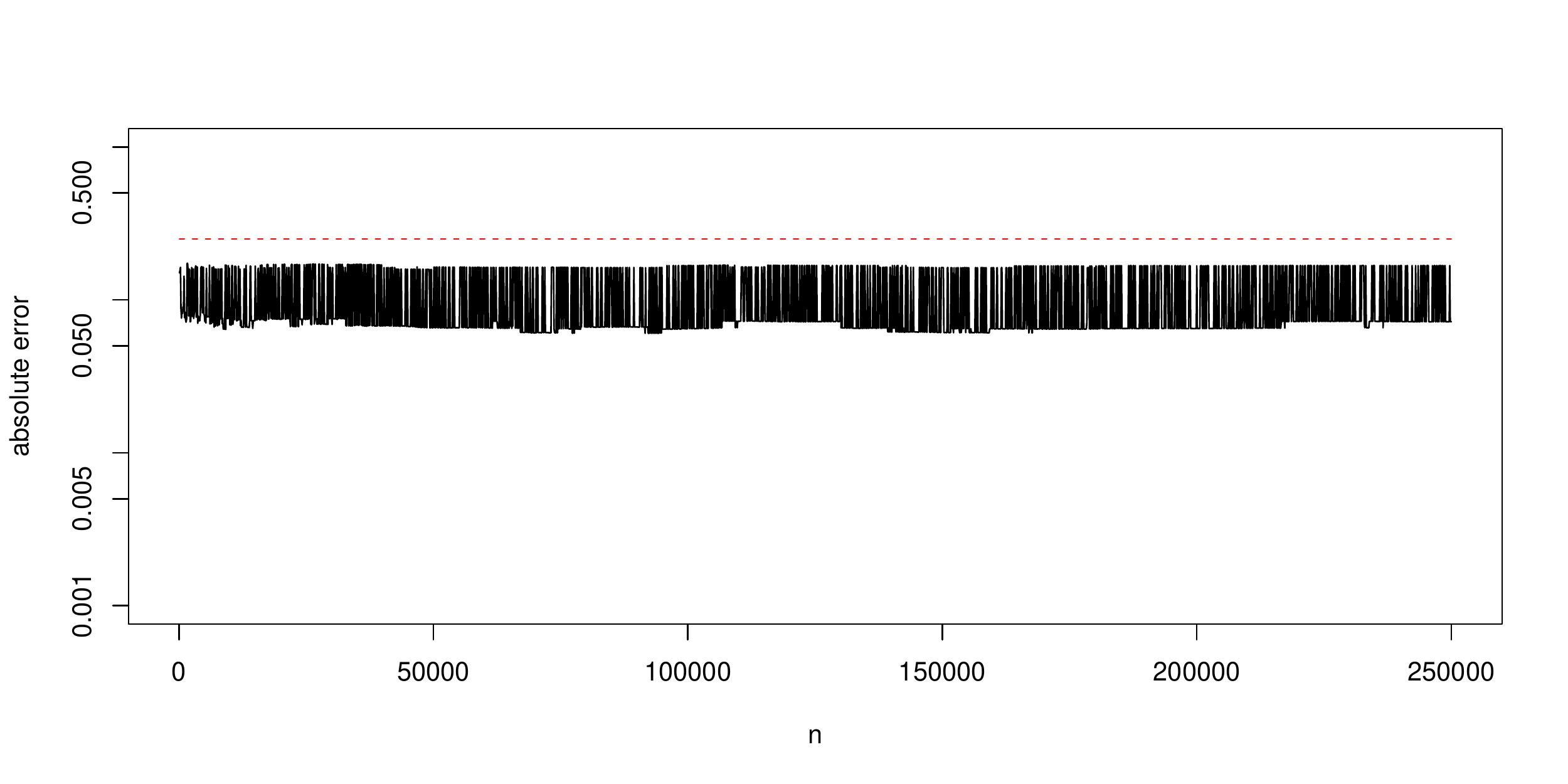}
  \caption{}
  \label{fig:error_copula_summary}
\end{subfigure}
\caption{The size in bytes of the copula summary (a) and the error $\abs{\hat{C}_{S}(0.7,0.7)-\hat{C}(0.7,0.7)}$ (b) for a data stream sampled from a bivariate Gaussian distribution with correlation $\rho=-0.8$. Here $\epsilon=0.05$ is used and the error bound of $5\epsilon$ is shown by the dashed red line.}
\end{figure}

The copula summary proposed in this paper which produces guaranteed error estimates to an empirical bivariate copula function for streaming data is now explored. Consider the following random stream of data, $\big\{x_{(1)}^{i},x_{(2)}^{i}\big\}_{i=1}^{n}$, where $x_{(1)}^{i} \sim N(0,1)$, $x_{(2)}^{i} \sim N(0,1)$ and $\rho(x_{(1)}^{i},x_{(2)}^{i})=-0.8$. Here, $\rho$ is the Pearson's correlation coefficient. For the first experiment $\epsilon=0.05$ is used alongside the algorithm in Sec.~\ref{sec:algorithm} and $n=2.5 \times 10^5$ is the length of the stream. The copula summary is constructed, with the size (in bytes) of the summary and the quantity $|\hat{C}_{S}(0.7,0.7)-\hat{C}(0.7,0.7)|$ being computed after every 100 elements are added to the stream. Figures \ref{fig:size_copula_summary} and \ref{fig:error_copula_summary} shows these two quantities over time respectively. The total space-memory used by the copula summary appears to be independent of $n$, beating the worst-case rate presented in Sec.~\ref{sec:error}. This suggests that the methodology presented in this paper could be used to compute bivariate empirical copula approximations indefinitely for any sized data streams without increasing the space-memory used. The theoretical error bound in Sec.~\ref{sec:error} is also shown in Figure \ref{fig:error_copula_summary}. Figure \ref{fig:subsummary_length} shows the length of the subsummaries $S_{(2)}^{i}$, for $i=1,...,L$ within the copula summary at the end of the stream. There are a few subsummaries that only contain one element, and that correspond to the elements in $S_{(1)}$ that have been added to the stream since the last use of the combine operation (see \ref{sec:appendixcombine}). Figure \ref{figure:time_copula_summary} shows the runtime (in seconds) for each 100'th iteration of the copula summary construction using a HP personal laptop. These stay relatively constant over the stream, and suggest that the algorithm can be used to compute copula approximations for any stream of data with an acquisition rate less than or equal to these runtimes. Such rates are seen in many applications of wireless sensing.

\begin{figure}
\centering
\begin{minipage}{.45\textwidth}
  \centering
  \vspace{-7mm}
  \includegraphics[width=65mm]{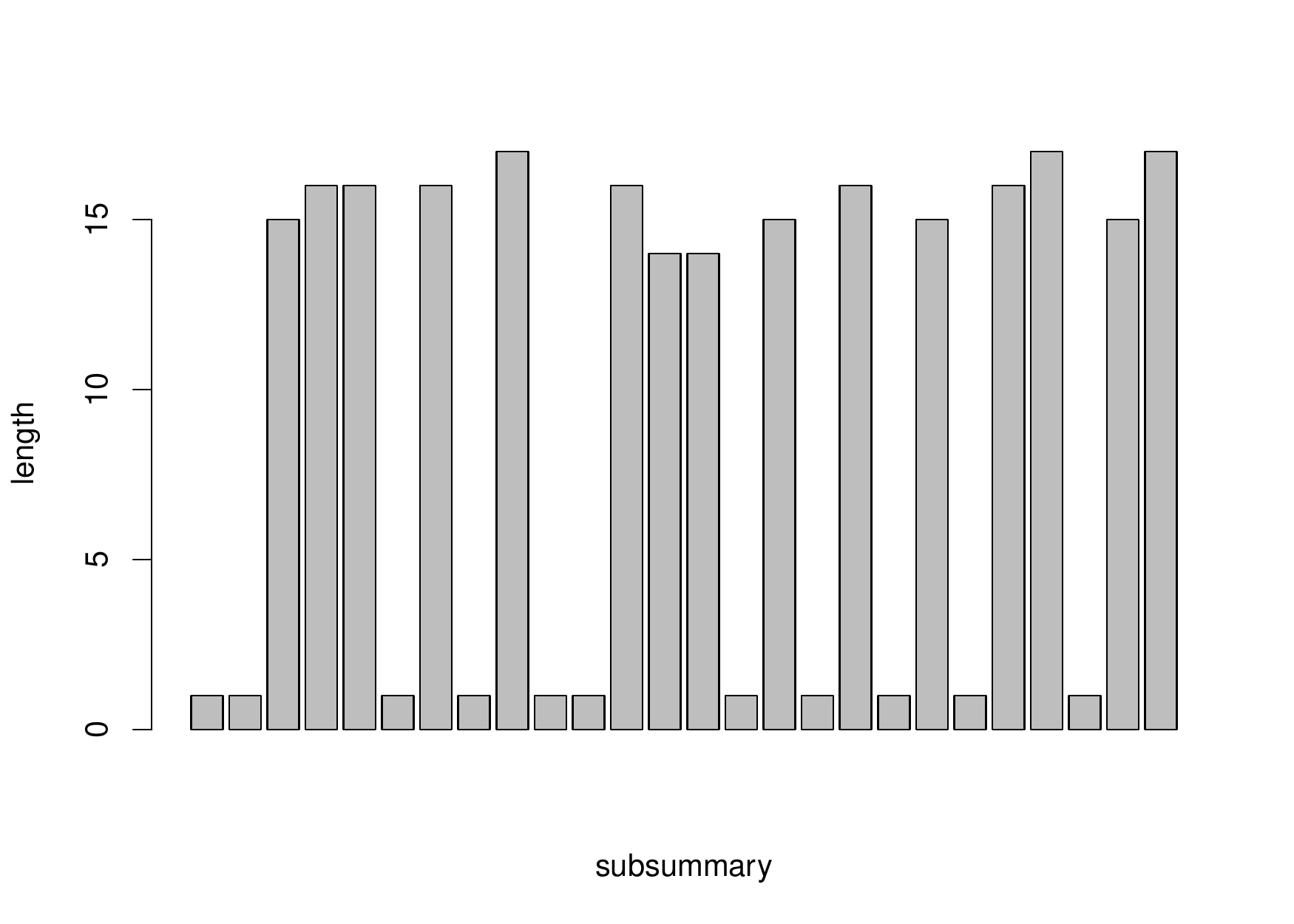}
\captionof{figure}{The length of all subsummaries $S_{(2)}^{i}$, for $i=1,...,L$, within the copula summary for a $2.5 \times 10^5$-element data stream sampled from a bivariate Gaussian distribution with correlation $\rho=-0.8$.}
\label{fig:subsummary_length}
\end{minipage}%
\hspace{2mm}
\begin{minipage}{.45\textwidth}
  \centering
\includegraphics[width=80mm]{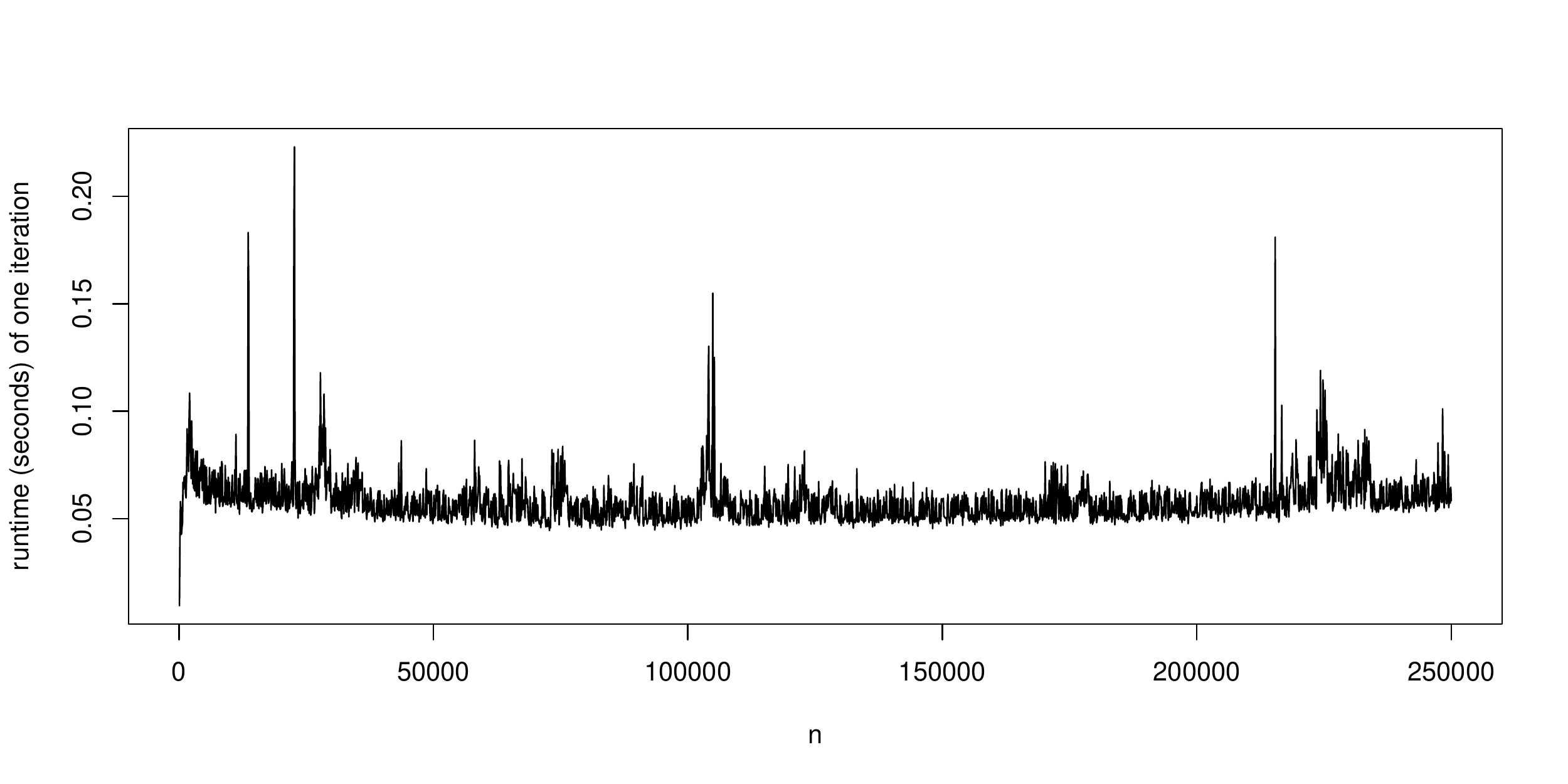}
\captionof{figure}{The runtime (in seconds) for each 100'th iteration of the copula summary construction, over a $2.5 \times 10^5$-element data stream sampled from a bivariate Gaussian distribution with correlation $\rho=-0.8$.}
\label{figure:time_copula_summary}
\end{minipage}
\end{figure}

\begin{figure}[t!]
  \centering
  \includegraphics[width=0.55\textwidth]{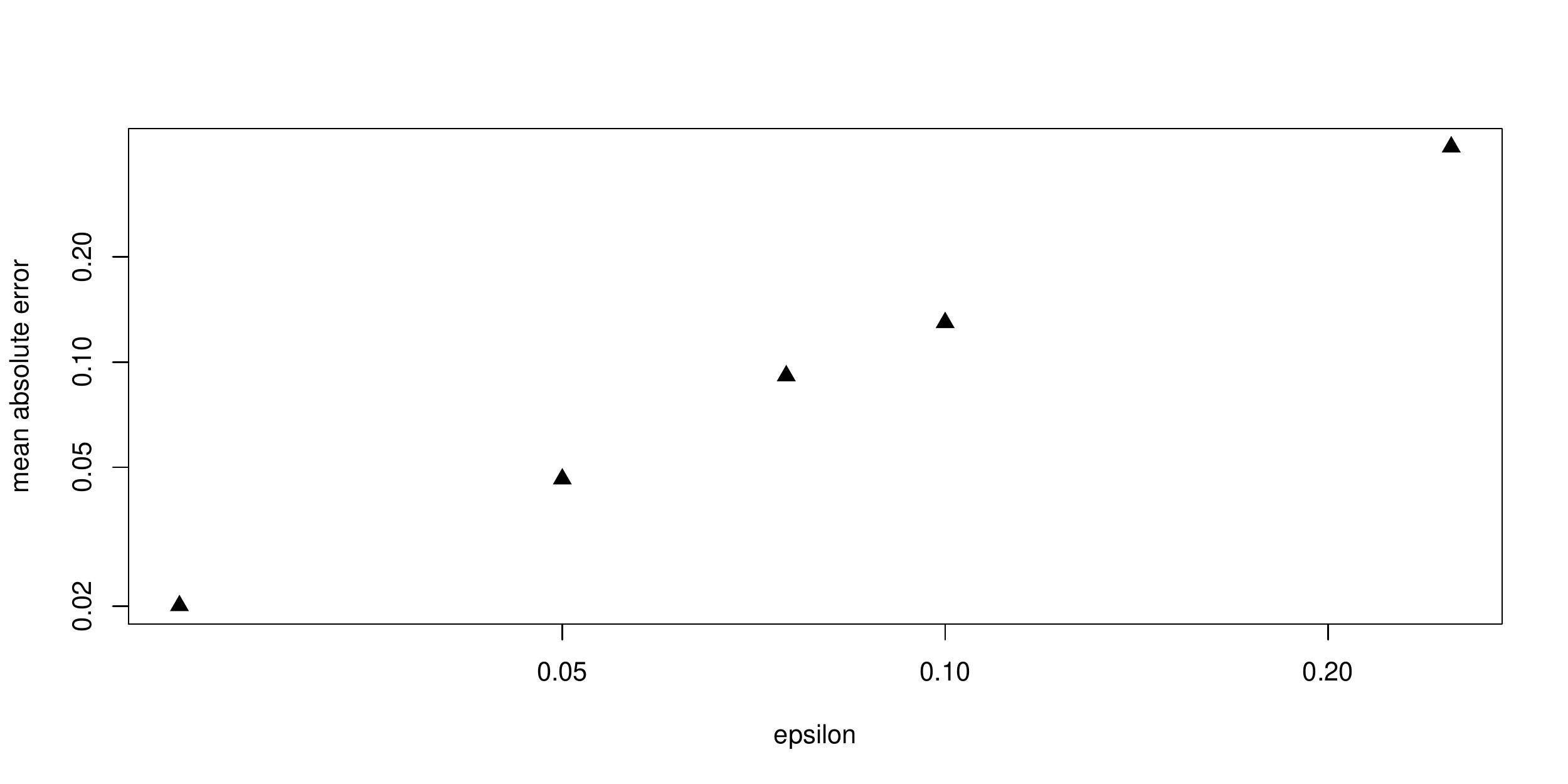}
  \caption{The average error $\abs{\hat{C}_{S}(0.7,0.7)-\hat{C}(0.7,0.7)}$ of copula summaries, over a range of $\epsilon$ values, for five independent data streams sampled from a bivariate Gaussian distribution with correlation $\rho=-0.8$.}
  \label{fig:epsilon_error_copula_summary}
\end{figure}

The approximation to the three dimensional ($d=3$) empirical copula decomposition in (\ref{equation:empiricaldecomposition}) will now be explored. Consider the stream of data $\big\{x_{(1)}^i,x_{(2)}^i,x_{(3)}^i\big\}_{i=1}^n$, where $\big\{x_{(1)}^{i}\big\}_{i=1}^{n}$ and $\big\{x_{(2)}^{i}\big\}_{i=1}^{n}$ are sampled from $N(0,1)$ with correlation $\rho=0.5$. The data $\big\{x_{(3)}^{i}\big\}_{i=1}^{n}$ are also sampled from $N(0,1)$, and have correlation with $\big\{x_{(2)}^{i}\big\}_{i=1}^{n}$ of $\rho=0.5$. The correlation between $\big\{x_{(1)}^{i}\big\}_{i=1}^{n}$ and $\big\{x_{(3)}^{i}\big\}_{i=1}^{n}$ is $\rho=0$. The steps in the previous section are implemented in order to compute an approximation to $\hat{C}(u_1,u_2,u_3)$ given in (\ref{equation:empiricaldecomposition}) for $d=3$. The copula summaries used to approximate the unconditional bivariate copulas $\hat{C}(u_1,u_2)$ and $\hat{C}(u_2,u_3)$ use $\epsilon=0.05$. First, Figure \ref{figure:pseudo_observation_error} demonstrates the error between the pseudo-observations $\big\{\hat{u}_{1|2}^i,\hat{u}_{3|2}^i\big\}_{i=1}^{n_{\text{query}}}$ obtained from the unconditional empirical copulas $\hat{C}(u_1,u_2)$ and $\hat{C}(u_2,u_3)$ and their bivariate copula summary approximations respectively. Given a trapezoidal approximation to the integrals in (\ref{equation:empiricalh}), the error of these pseudo-observations should be bounded by that of a single query from an $\epsilon$-approximate copula summary, $5\epsilon$; this indeed is the case. Second, these pseudo-observations are used to generate the pseudo-data in (\ref{equation:pseudodatasummary}) and thus an approximation to the conditional empirical copula $\hat{C}(u_{1|2}, u_{3|2})$. Finally the decomposition in (\ref{equation:summarydecomposition}) can be computed. This approximation is evaluated at both $u_2=0.1$ and $u_2=0.9$ for a grid of $(u_1,u_3)$ values in Figures \ref{fig:copula_3d_approximate_marg_01} and \ref{fig:copula_3d_approximate_marg_09} respectively. In Figures \ref{fig:copula_3d_empirical_marg_01} and \ref{fig:copula_3d_empirical_marg_09} the empirical decomposition in (\ref{equation:empiricaldecomposition}) evaluated at the same values of $(u_1,u_2,u_3)$ is shown for comparison.

\begin{figure}[t!]
  \centering
  \includegraphics[width=0.8\textwidth]{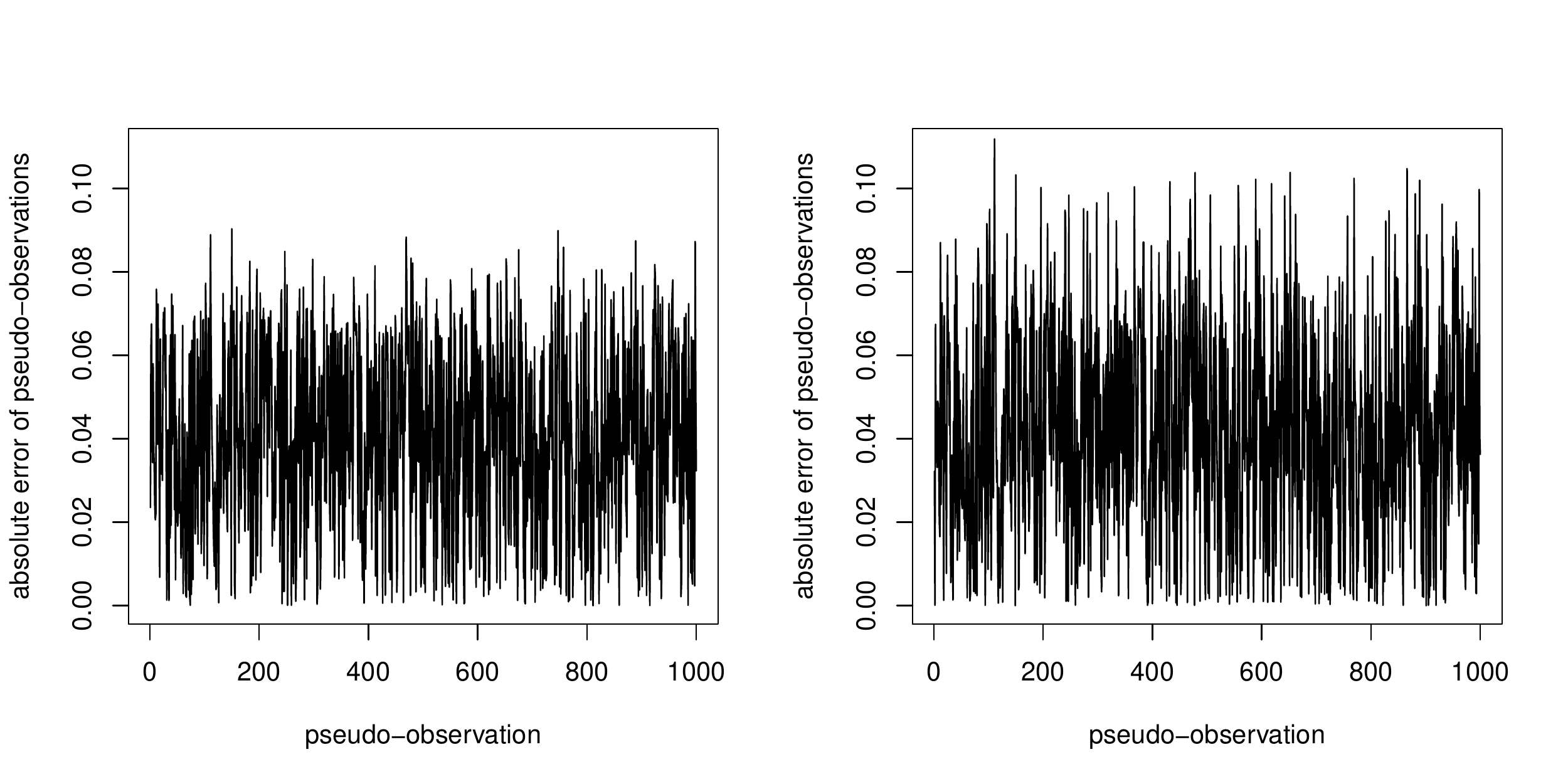}
  \caption{The absolute error between the pseudo-observations $\big\{\hat{u}_{1|2}^i\big\}_{i=1}^{n_{\text{query}}}$ (left) and $\big\{\hat{u}_{3|2}^i\big\}_{i=1}^{n_{\text{query}}}$ (right) computed from the empirical copulas and the copula summary approximation for $\hat{C}(u_1,u_2)$ and $\hat{C}(u_2,u_3)$ respectively.}
  \label{figure:pseudo_observation_error}
\end{figure}

\begin{figure}[t!]
\centering
\begin{subfigure}{0.5\textwidth}
  \centering
  \includegraphics[width=0.8\textwidth]{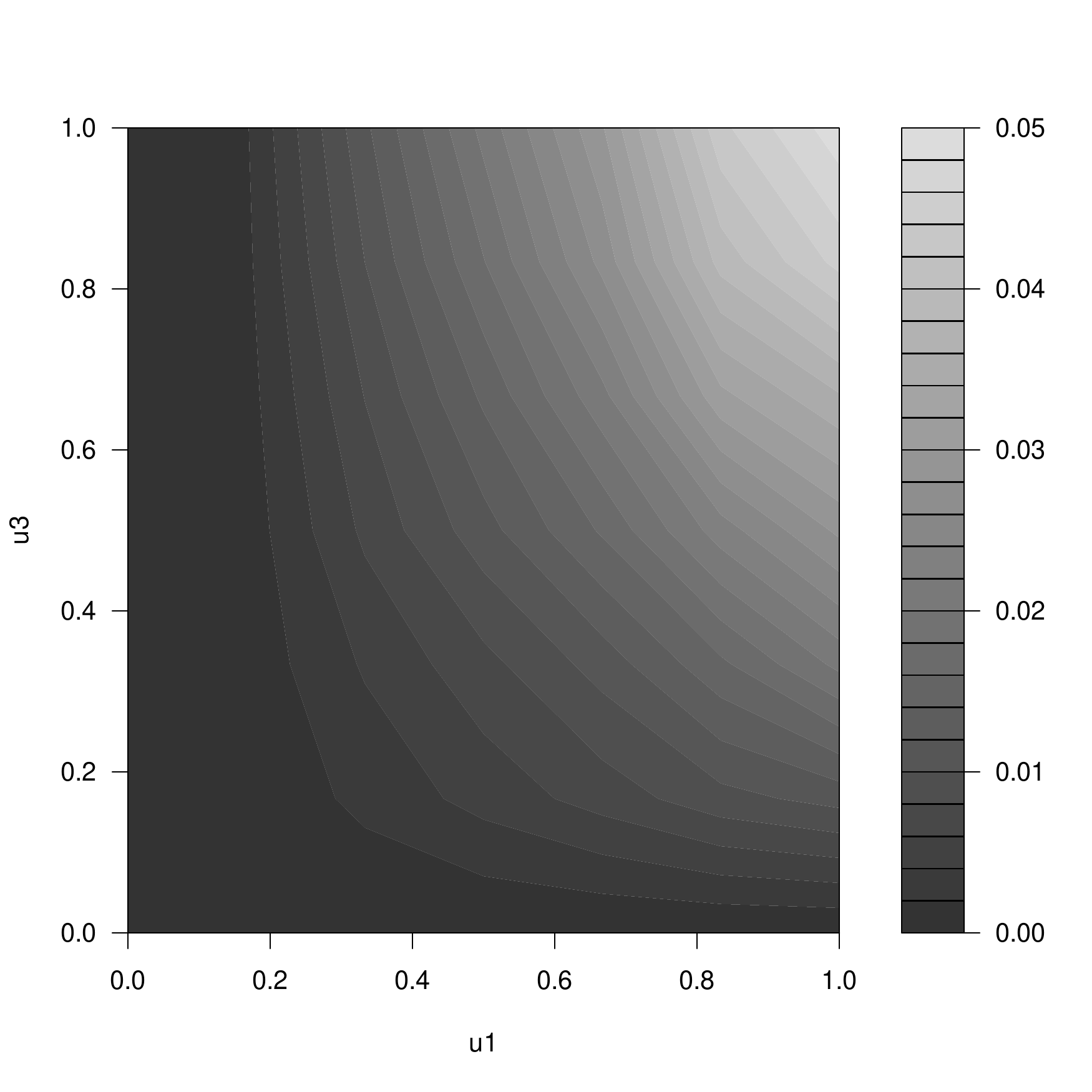}
  \caption{}
  \label{fig:copula_3d_approximate_marg_01}
\end{subfigure}%
\begin{subfigure}{0.5\textwidth}
\centering
  \includegraphics[width=0.8\textwidth]{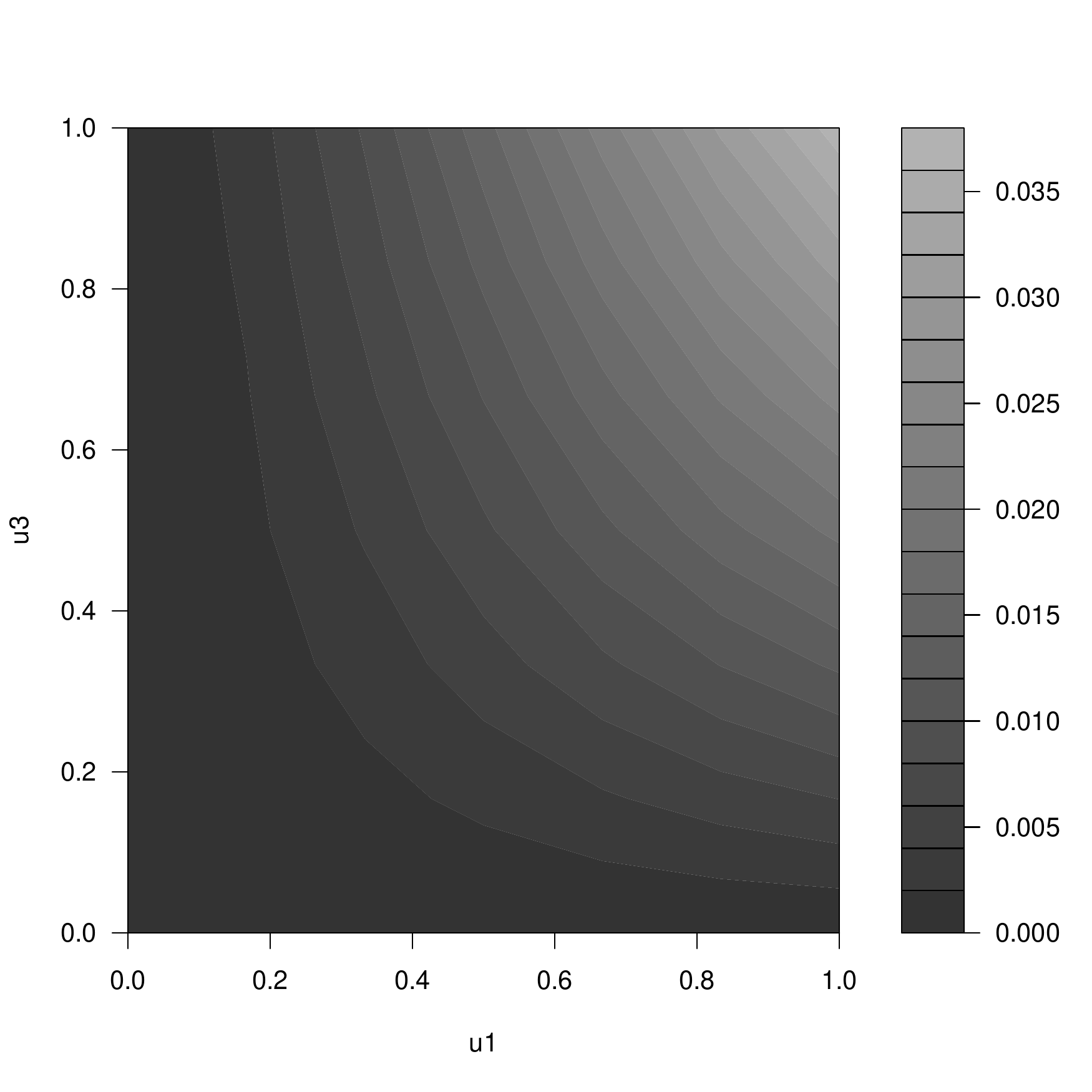}
  \caption{}
  \label{fig:copula_3d_empirical_marg_01}
\end{subfigure}
\caption{The approximation $\hat{C}_{S}(u_1,0.1,u_3)$ (a), and the empirical decomposition $\hat{C}(u_1,0.1,u_3)$ (b) evaluated at a grid of ($u_1,u_3)$ values.}
\end{figure}

\begin{figure}[t!]
\centering
\begin{subfigure}{0.5\textwidth}
  \centering
  \includegraphics[width=0.8\textwidth]{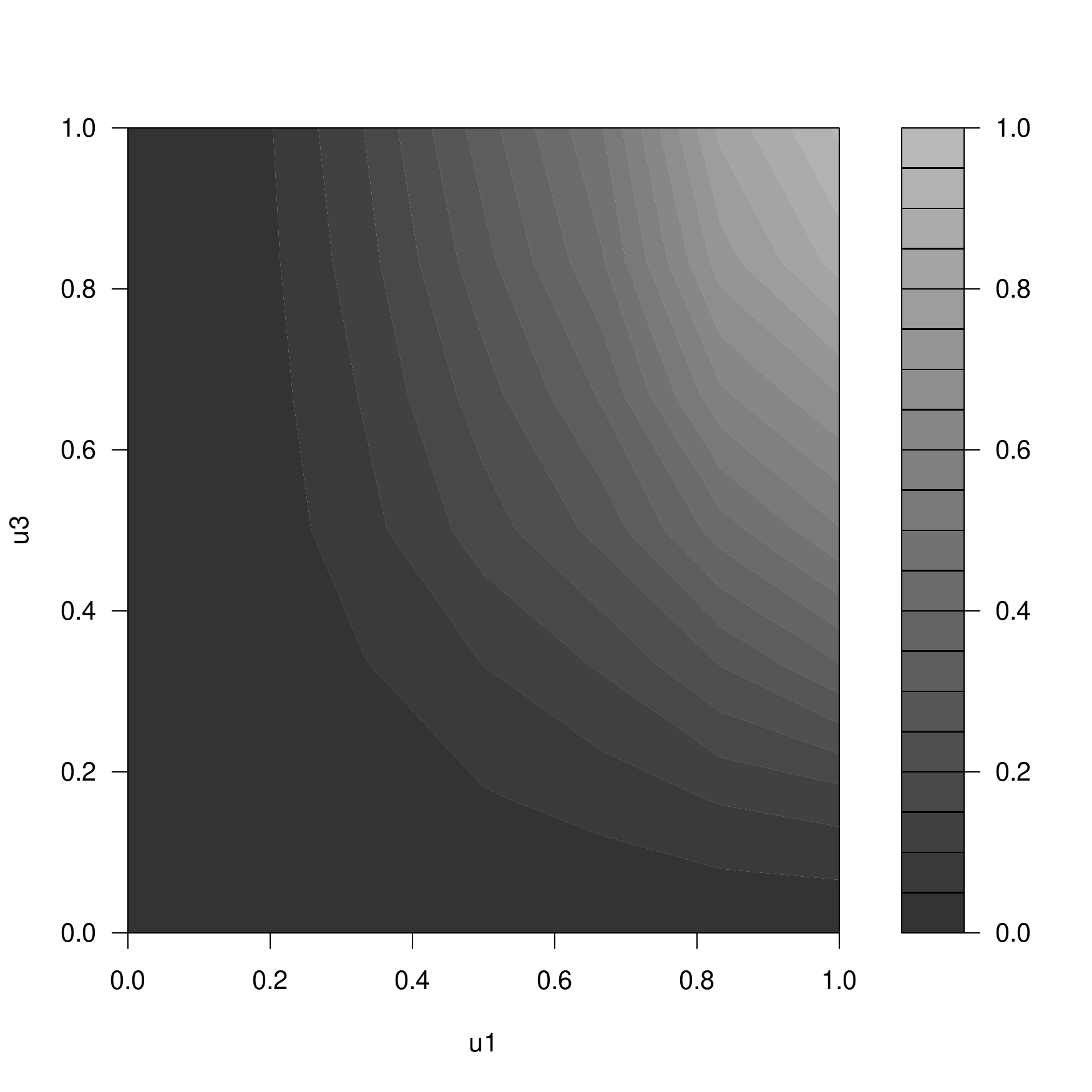}
  \caption{}
  \label{fig:copula_3d_approximate_marg_09}
\end{subfigure}%
\begin{subfigure}{0.5\textwidth}
\centering
  \includegraphics[width=0.8\textwidth]{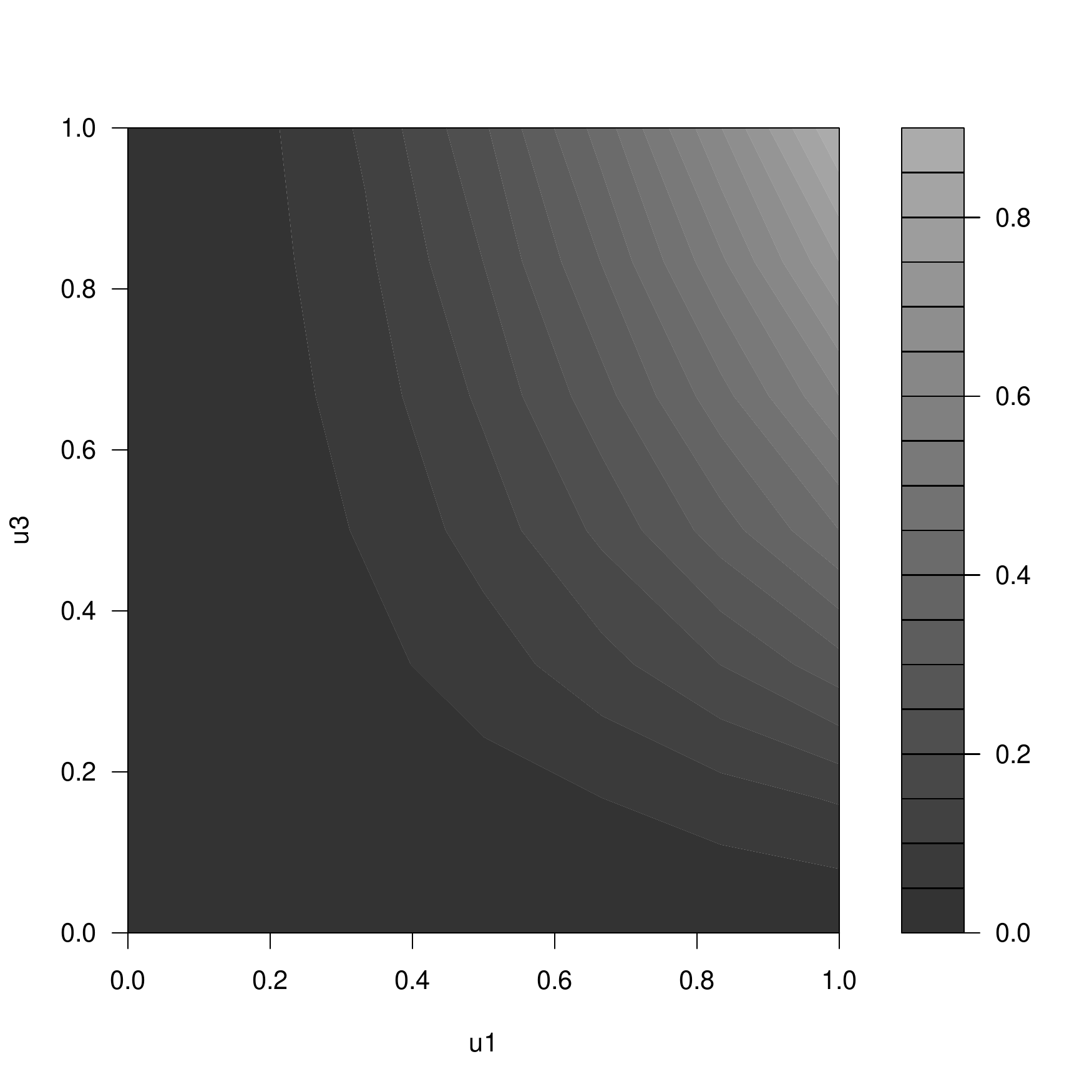}
  \caption{}
  \label{fig:copula_3d_empirical_marg_09}
\end{subfigure}
\caption{The approximation $\hat{C}_{S}(u_1,0.9,u_3)$ (a), and the empirical decomposition $\hat{C}(u_1,0.9,u_3)$ (b) evaluated at a grid of ($u_1,u_3)$ values.}
\end{figure}

\section{Conclusion}

This paper has proposed an algorithm to approximate an empirical copula function of a bivariate data stream with a space-memory constraint. These approximations have a guaranteed error bound that has been presented here, and can be used to model the dependence structure between two streams of data. These approximations could also be used as a tool to construct approximations for higher dimensional copulas in the streaming data context, when they are given as a decomposition containing bivariate empirical copulas. This paper gives an example of how one may achieve this.
A natural extension of this work is to use these approximations to derive estimates of rank correlation coefficients with guaranteed error bounds, such as the Kendall Tau correlation coefficient \citep{Xiao}.

This algorithm is a generalisation to the one dimensional quantile summaries constructed via the Greenwald and Khanna algorithm \citep{Greenwald}. The data structure is similar to the ones used to find multidimensional ranges in \cite{Suri} and \cite{Hershberger}. It is formed via a particular combination of $L+1$ different $\epsilon$-approximate quantile summaries and therefore the algorithm uses a worst-case $\mathcal{O}\left(\frac{1}{\epsilon^2}\log(\epsilon n)^2\right)$ space-memory after $n$ elements in the bivariate stream data. Numerical experiments in this paper have confirmed the space-memory efficiency and the theoretical error bound of the approximation.


\section{Acknowledgements}

The author would like to kindly thank the reviewers and the editorial team for their great help during the peer-review process. This work was supported by The Alan Turing Institute under the EPSRC grant EP/N510129/1 and the Turing-Lloyd's Register Foundation Programme for Data-Centric Engineering.


\appendix

\section{Modifications to the Greenwald and Khanna algorithm}

This appendix section will cover in detail the operations used in the construction and querying of the copula summary proposed in this paper. The operations in \ref{sec:appendixinsert}, \ref{sec:appendixcombine} and \ref{sec:merging} are slightly modified versions of the ones presented in \cite{Greenwald} and \cite{GreenwaldMerge}.


\subsection{The insert and combine operations}

\subsubsection{Insert}
\label{sec:appendixinsert}

When an element $(x_{(1)}^{n+1},x_{(2)}^{n+1})$ gets added to the bivariate stream $\big\{x_{(1)}^j,x_{(2)}^j\big\}_{j=1}^n$ the following operation occurs:

\begin{enumerate}
\item Define the following new subsummary $S_{(2)}^{*}=\big\{(x_{(2)}^{n+1},1,0)\big\}$.
\item If $x_{(1)}^{n+1}<v_{1}$, then input $(x_{(1)}^{n+1},1,0)$ at the start of $S_{(1)}$ and place $S_{(2)}^{*}$ into the copula summary so that it becomes $\big\{S_{(1)},S_{(2)}^*,S_{(2)}^1,\ldots,S_{(2)}^L\big\}$. Conversely, if $x_{(1)}^{n+1}\geq v_{L}$, then input $(x_{(1)}^{n+1},1,0)$ at the end of $S_{(1)}$  and place $S_{(2)}^{*}$ into the copula summary so that it becomes $\big\{S_{(1)},S_{(2)}^1,\ldots,S_{(2)}^L,S_{(2)}^*\big\}$.
\item Otherwise, find $i$ where $v_{i} \leq x_{(1)}^{n+1} < v_{i+1}$. Then compute $\Delta_{(1)}^{*}=g_{(1)}^{i+1}+\Delta_{(1)}^{i+1}-1$, and insert $(x_{(1)}^{n+1},1,\Delta_{(1)}^{*})$ into $S_{(1)}$ in between $(v_{i},g^{i}_{(1)},\Delta^{i}_{(1)})$ and $(v_{i+1},g^{i+1}_{(1)},\Delta^{i+1}_{(1)})$. Place $S_{(2)}^{*}$ into the copula summary in between $S_{(2)}^i$ and $S_{(2)}^{i+1}$ so that the updated summary becomes $\big\{S_{(1)},S_{(2)}^1,\ldots,S_{(2)}^i,S_{(2)}^*,S_{(2)}^{i+1},\ldots,S_{(2)}^{L}\big\}$.
\end{enumerate}

\subsubsection{Combine}
\label{sec:appendixcombine}

The combine operation occurs every time $n$ is divisible by $\floor*{1/(2\epsilon)}$. Tuples in an $\epsilon$-approximate summary can be combined to remove unnecessary tuples. By definition, one can find the number of elements in the stream so far at any time by computing $n=\sum^{L}_{i=1}g_{(1)}^{i}$. When the following operation is implemented on $S_{(1)}$, the subsummaries associated with each combined tuple in $S_{(1)}$ also need to be merged. The combine operation is implemented via the method below, whenever $L \geq 3$. Start with $j=L$.

\begin{enumerate}

\item Set $k=j$. If $g_{(1)}^{j}+\Delta_{(1)}^{j}<2\epsilon n$, decrease $k$ by one at a time assuring $\sum^{k}_{i=0}g_{(1)}^{j-i}+\Delta_{(1)}^{j}<2 \epsilon n$. When this is no longer possible, or $j-k=2$, stop decreasing $k$.

\item Merge the $\epsilon$-approximate subsummaries $S_{(2)}^{j-k}$,...,$S_{(2)}^{j}$ to form the new $\epsilon$-approximate subsummary $Q=M(S_{(2)}^{j-k},...,S_{(2)}^{j})$. The merge operation is explained in the next appendix section.

\item Replace the tuples
$$
(v_{j-k},g_{(1)}^{j-k},\Delta_{(1)}^{j-k}),...,(v_{j},g_{(1)}^{j},\Delta_{(1)}^{j})
$$
by the tuple $(v_{j},\sum^{k}_{i=0}g_{(1)}^{j-i},\Delta_{(1)}^{j})$. Also replace the subsummaries $(S_{(2)}^{j-k},\ldots,S_{(2)}^{j})$ with $Q$ in the copula summary so that the updated summary becomes $\big\{S_{(1)},S_{(2)}^{1},\ldots,S_{(2)}^{j-k-1},Q,S_{(2)}^{j+1},\ldots,S_{(2)}^{L}\big\}$.

\item Finally, the tuples $(w,g_{(2)},\Delta_{(2)})$ within the subsummary $Q$ are also combined in the same way as steps (1), (3) and (5) to remove unnecessary tuples.

\item Set $j=j-k$, then go back to step (1) if $j > 2$.

\end{enumerate}

The aim of this operation is to simultaneously refine all $L+1$ $\epsilon$-approximate quantile summaries in the copula summary, limiting the space-memory used. By requiring that the first tuple in each summary/subsummary is not combined with any other tuple during this operation, the algorithm preserves the smallest and largest elements seen in both $\big\{x_{(1)}^{i}\big\}_{i=1}^{n}$ and $\big\{x_{(2)}^{i}\big\}_{i=1}^{n}$ individually.

\subsection{Merging quantile summaries}
\label{sec:merging}

The merge operation for merging quantile summaries, introduced in \cite{GreenwaldMerge}, is now explained. This is utilised during step (2) of the combine operation explained in the previous appendix section. Suppose we merge the $\epsilon$-approximate summaries $Q_1$, of length $L_{Q_1}$, and $Q_2$, of length $L_{Q_2}$, to obtain the summary $M(Q_1,Q_2)$. This summary is also $\epsilon$-approximate, and is of length $L_{Q_1}+L_{Q_2}$ after the merge. Suppose $M(Q_1,Q_2)$ has the elements $Q_1 \bigcup Q_2$ and $w_k$ is an element in $M(Q_1,Q_2)$ from $Q_1$. Let $w_1$ be the largest element (if it exists) in $Q_2$ that is less than or equal to $w_k$. Let $w_2$ be the smallest element (if it exists) in $Q_2$ that is greater than $w_k$. Then the parameters $r_{max,M(Q_1,Q_2)}(w_k)$ and $r_{min,M(Q_1,Q_2)}(w_k)$ are given by
\begin{equation}
r_{min,M(Q_1,Q_2)}(w_k)=
\begin{cases}
r_{min,Q_2}(w_1) + r_{min,Q_1}(w_k), & \text{if } w_1 \text{ exists}\\
r_{min,Q_1}(w_k), & \text{otherwise},
\end{cases}
\end{equation}
and
\begin{equation}
r_{max,M(Q_1,Q_2)}(w_k)=
\begin{cases}
r_{max,Q_2}(w_2) + r_{max,Q_1}(w_k) - 1, & \text{if } w_2 \text{ exists}\\
r_{max,Q_2}(w_1) + r_{max,Q_1}(w_k), & \text{otherwise}.
\end{cases}
\end{equation}
In the implementations of the merge operation used in this paper, we treat a merge of a summary, $R$, with only one tuple $(v,g,\Delta)=(v,1,0)$ and a summary $Q$ of arbitrary length slightly differently to this. In this case, the insert operation from \cite{Greenwald} is implemented on $(v, g, \Delta)$, into the summary $Q$. The merge operation can also be used recursively: the summary $M(Q_{1},Q_{2},Q_{3})$ for example can be constructed by merging $M(Q_{1},Q_{2})$ and $Q_{3}$.

\subsection{Querying a quantile summary}

\label{sec:appendixquery}

The following method can be used to query a single $\epsilon$-approximate summary $Q$ for the $u$-quantile. Let $r=\ceil{u n}$. This is the rank of the element in the stream $\big\{x^{i}\big\}_{i=1}^{n}$ that one would like to approximate the value of.
\begin{enumerate}
\item Compute $\big\{r_{max,Q}(v_i)\big\}_{i=1}^{L}$.
\item If $r \geq n - \floor{\epsilon n}$, then return $\tilde{F}_{n}^{-1}(u)=v_L$.
\item If $r < n - \floor{\epsilon n}$, then return $\tilde{F}_{n}^{-1}(u)=v_j$, where $r_{max,Q}(v_j)$ is the smallest element in $\big\{r_{max,Q}(v_i)\big\}_{i=1}^{L}$ that is greater than $r+\floor{\epsilon n}$.
\end{enumerate}
The work in \cite{Greenwald} showed that $j \in [r-\epsilon n,r+\epsilon n]$, for the $j$ that satisfies $\tilde{F}_{n}^{-1}(u)=\tilde{x}^{j}$.


\subsection{Inversely querying a quantile summary}

\label{sec:appendixinversequery}

The following method can be used to find an approximation to the rank $r$ from an $\epsilon$-approximate quantile summary $Q$, where $\tilde{x}^{r}$ is the largest element in a stream $\big\{x^{i}\big\}_{i=1}^{n}$ that is at most $y$. In other words, to approximate $n\hat{F}_{n}(y)$, where $\hat{F}_{n}(y)$ is the empirical CDF of the stream. 
\begin{enumerate}
\item If $y\geq v_L$, let $r=n$. Conversely, if $y < v_1$, let  
$r=0$.
\item Otherwise, find $i$ where $v_i \leq y < v_{i+1}$, and then set $r=r_{max,Q}(v_{i})$.
\item Finally output $\tilde{F}_{n}(y)=r/n$.
\end{enumerate}
Denote this approximation by $\tilde{F}_{n}(y)$. The work in \cite{Lall} showed that this method obtains an approximation within the interval $[\hat{F}_{n}(y)-3\epsilon,\hat{F}_{n}(y)+3\epsilon]$.

\subsection{Querying the copula summary}

\label{sec:appendixcopulaquery}

Using the queries in \ref{sec:appendixquery} and \ref{sec:appendixinversequery}, the copula summary can be queried to return an approximation to $\hat{C}(u_1,u_2)$ via the following method.

\begin{enumerate}
\item Query the $\epsilon$-approximate summary $S_{(1)}^{i}$ for the $r=\ceil*{u_1n}$ ranked element, where $n$ is the number of elements so far in the stream, using \ref{sec:appendixquery}.

\item Merge the $\epsilon$-approximate subsummaries $\big\{S_{(2)}^{i}\big\}_{i=1}^{E}$ to form $P_1=M(S_{(2)}^{1},...,S_{(2)}^{E})$ and merge the $\epsilon$-approximate subsummaries $\big\{S_{(2)}^{i}\big\}_{i=1}^{L}$ to form $P_2=M(S_{(2)}^{1},...,S_{(2)}^{L})$. Here, $E$ is defined in (\ref{equation:Edefine}).

\item Query the $\epsilon$-approximate summary $P_2$ for the $\ceil*{u_2n}$ ranked element, and denote this query by $\tilde{F}_{n,(2)}^{-1}(u_2)$.

\item Inversely query the $\epsilon$-approximate summary $P_1$ by $\tilde{F}_{n}^{-1}(u_2)$, and denote this query by $\tilde{F}_{\hat{n}_1,(2)}\left(\tilde{F}^{-1}_{n,(2)}(u_2)\right)$.

\item Then define the overall copula query as
\begin{equation}
\hat{C}_{S}(u_1,u_2) = \frac{\hat{n}_{1}}{n}\tilde{F}_{\hat{n}_1,(2)}\left(\tilde{F}^{-1}_{n,(2)}(u_2)\right),
\label{equation:copulaquery}
\end{equation}
where $\hat{n}_{1}$ is defined in (\ref{equation:hatndefine}).
\end{enumerate}



\bibliographystyle{elsarticle-harv} 
\bibliography{refs}





\end{document}